\documentclass[aps,pra,superscriptaddress,nofootinbib,twocolumn]{revtex4-2}

\usepackage{amsmath,amsthm,amsfonts,amssymb,amscd} 
\usepackage[utf8]{inputenc}
\usepackage{indentfirst}
\usepackage{graphicx}
\usepackage{color}
\usepackage[T1]{fontenc}
\usepackage[american,ngerman,greek,brazilian,british]{babel}
\usepackage{float}
\usepackage[dvipsnames]{xcolor}
\usepackage[colorlinks=true,citecolor=Mulberry,linkcolor=prettygreen,urlcolor=MidnightBlue,hyperindex,breaklinks]{hyperref}
\usepackage{breakurl}
\usepackage{cleveref}
\usepackage{mathtools,xfrac}
\usepackage{bm}
\usepackage{ulem}
\usepackage{enumerate}
\usepackage{sidecap}
\usepackage{microtype}
\usepackage{enumitem}
\usepackage{bbold}
\usepackage{listings}
\usepackage{algorithm}
\usepackage{algpseudocode}
\usepackage{multirow}
\usepackage{setspace}
\usepackage{physics}
\usepackage{dsfont}
\usepackage{soul}
\usepackage{lipsum}  
\usepackage{comment}

\definecolor{awesome}{rgb}{0.93, 0.53, 0.18}
\definecolor{prettygreen}{RGB}{5,125,143}

\newcommand{\ot}{\otimes}

\newtheorem{theorem}{Theorem}
\newtheorem*{theorem*}{Theorem}
\newtheorem{lemma}{Lemma}
\newtheorem{corollary}{Corollary}

\begin{document}

\title{Catalytic Activation of Bell Nonlocality}

\author{Jessica Bavaresco}
\affiliation{Department of Applied Physics, University of Geneva, Geneva, Switzerland}
\affiliation{Sorbonne University, CNRS, LIP6, F-75005 Paris, France}

\author{Nicolas Brunner}
\affiliation{Department of Applied Physics, University of Geneva, Geneva, Switzerland}

\author{Antoine Girardin}
\affiliation{Department of Applied Physics, University of Geneva, Geneva, Switzerland}

\author{Patryk Lipka-Bartosik}
\affiliation{Department of Applied Physics, University of Geneva, Geneva, Switzerland}
\affiliation{Center for Theoretical Physics, Polish Academy of Sciences, Warsaw, Poland}
\affiliation{Institute of Theoretical Physics, Jagiellonian University, 30-348 Kraków, Poland}

\author{Pavel Sekatski}
\affiliation{Department of Applied Physics, University of Geneva, Geneva, Switzerland}

\begin{abstract} 
The correlations of certain entangled states can be perfectly simulated classically via a local model. Hence such states are termed Bell local, as they cannot lead to Bell inequality violation. Here, we show that Bell nonlocality can nevertheless be activated for certain Bell-local states via a catalytic process. Specifically, we present a protocol where a Bell-local state, combined with a catalyst, is transformed into a Bell-nonlocal state while the catalyst is returned exactly in its initial state. Importantly, this transformation is deterministic and based only on local operations. Moreover, this procedure is possible even when the state of the catalyst is itself Bell local, demonstrating a new form of superactivation of Bell nonlocality, as well as an interesting form of quantum catalysis. On the technical level, our main tool is a formal connection between catalytic activation and many-copy activation.
\end{abstract}

\maketitle

\section{Introduction}

Quantum entanglement enables strong forms of correlations, which have no equivalent in classical physics. This is the effect of quantum Bell nonlocality: by performing well-chosen local measurements on some entangled state, one obtains correlations that cannot be reproduced within any physical theory satisfying a natural condition of locality, as formalized by Bell~\cite{bell1964einstein,Brunner_2014}. Beyond the purely conceptual interest, these ideas have played a significant role in quantum information processing, notably for device-independent applications~\cite{Acin2007,Pironio2010,colbeck2011,supic2020self}.

From a fundamental perspective, a long-standing problem has been to understand precisely the relation between entanglement and Bell nonlocality
---specifically, whether every entangled state can generate Bell nonlocality. While this is always the case for pure entangled states~\cite{GISIN1991}, the problem is much more complicated for mixed states~\cite{werner1989quantum,Barrett_2002,acin2006grothendieck,augusiak2015entanglement,Bowles2016}. Indeed, there exist mixed entangled states that are ``Bell local'', i.e., their correlations can be simulated via a purely classical model (involving a shared random variable) for arbitrary local measurements~\cite{werner1989quantum,Barrett_2002}.

To make the question more intriguing, it was later realized that Bell nonlocality can be ``activated’’ in some cases, via an adequate processing. Of course, to ensure a fair accounting, it is critical that this processing alone cannot generate Bell nonlocality. Starting from certain Bell-local states $\rho$, Bell nonlocality can be activated via an initial local filtering of $\rho$~\cite{popescu1995bell,gisin1996hidden,Masanes2008,hirsch2013genuine}, or by processing multiple copies of $\rho$ via joint local measurements~\cite{palazuelos2012superactivation,cavalcanti2013all,Contreras2022}, within a quantum network~\cite{sen2005entanglement,Cavalcanti2011,cavalcanti2012nonlocalitytests}, or also by broadcasting $\rho$ locally~\cite{bowles2021single,Boghiu2023}. At the moment, it remains unclear for which entangled states Bell nonlocality can be activated, and what is the relation between these different activation scenarios. A key question is whether there exists a scenario where it would be possible to generate Bell nonlocality starting from any entangled state $\rho$. If possible, this would demonstrate an operational equivalence between entanglement and Bell nonlocality.

Here we investigate a novel scenario for activating Bell nonlocality via quantum catalysis~\cite{review1,review2}, a concept that has been proven useful in quantum information processing~\cite{Jonathan_1999,vanDam2003,turgut2007catalytic}, quantum thermodynamics~\cite{Brand_o_2015,Ng_2015,Wilming2017} as well as other areas~\cite{Aberg2014,Vaccaro2018,lostaglio2019coherence,takagi2022correlation,char2023catalytic}. We consider two distant parties, Alice and Bob, sharing a system initially prepared in an entangled, yet Bell-local, quantum state $\rho_{AB}$. We show that, by combining the initial system with an ancillary bipartite system, the \textit{catalyst}, and by allowing Alice and Bob to locally manipulate their subsystems, it is possible to transform the original system from a Bell-local into a \textit{Bell-nonlocal} state in a catalytic manner---meaning that the catalyst is recovered in exactly the same state as it was in initially. Importantly, this catalytic transformation is deterministic, based only on local operations, and does not involve any classical communication between Alice and Bob. Moreover, we show that this phenomenon is possible even when the catalyst is itself in a Bell-local state. Hence, this procedure constitutes a new form of \textit{superactivation} of Bell nonlocality, as neither the system nor the catalyst could have initially led to Bell inequality violation. We connect this phenomenon with that of many-copy activation of Bell nonlocality~\cite{palazuelos2012superactivation}, by showing that whenever a Bell-local entangled state can be activated with many copies, it can also be catalytically activated.

\begin{figure*}
    \centering
    \includegraphics[width=0.9\linewidth]{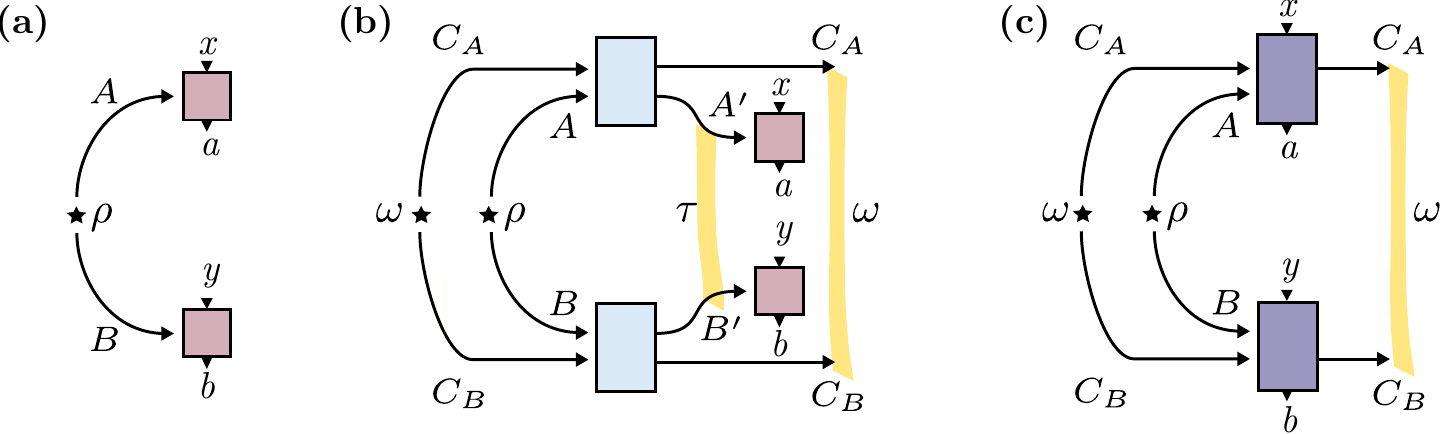}
    \caption{
    {\bf(a)} The standard Bell nonlocality scenario, where an entangled state $\rho_{AB}$ is shared between two parties, $A$ and $B$, who perform local measurements  to obtain the distribution $p(ab|xy)$. {\bf (b)} In this work, we discuss a scenario for catalytic activation of Bell nonlocality. Before being measured, the state $\rho_{AB}$ is catalytically transformed into $\tau_{A' B'}$ by means of deterministic local operations. That is, the joint state of the systems and the bipartite catalyst $\rho_{AB}\otimes \omega_{C_A C_B}$ is mapped to a global output state  $\tau_{A' B' C_A C_B}$ (by joint operations on $AC_A$ and $BC_B$), in such a way that the marginal state of the catalyst is returned unchanged $\tau_{C_A C_B}=\omega_{C_A C_B}$. We show that the output state of the measured systems $\tau_{A' B'}$ can be Bell nonlocal, even when starting from a state $\rho_{AB}$ that is Bell local, i.e., it cannot lead to nonlocality in scenario (a). {\bf (c)} A different variant of catalytic activation of Bell nonlocality, where the catalyst returned only after the local measurements.
    }
    \label{fig:1}
\end{figure*}

\section{Bell nonlocality}

The standard scenario investigated in quantum Bell nonlocality [see Fig.~\ref{fig:1}(a)] consists of two spatially separated and non-communicating parties, Alice and Bob, who share a bipartite quantum state $\rho_{AB}$. Each party can choose among different quantum measurements $\{M_A^{a|x}\}$ and $\{M_B^{b|y}\}$---labeled by classical values $x$ for Alice and $y$ for Bob---with which to act upon their subsystems, yielding classical outcomes $a$ and $b$, respectively. The statistics of such an experiment are described by a set of conditional joint probability distributions $\{p(ab|xy)\}$ given by
\begin{align}
    p(ab|xy) = \Tr[(M_A^{a|x} \ot M_B^{b|y})\, \rho_{AB}].
\end{align}

A state $\rho_{AB}$ is called \textit{Bell local} when, for all possible local quantum measurements $\{M_A^{a|x}\}$ and $\{M_B^{b|y}\}$, the corresponding distributions $\{p(ab|xy)\}$ can be described by a local hidden variable (LHV) model, according to
\begin{align}
    \label{eq:local}
    p(ab|xy) = \sum_\lambda \pi(\lambda) \, p_A(a|x,\lambda) \, p_B(b|y, \lambda),
\end{align}
for a shared LHV described by a distribution $\{\pi(\lambda)\}$ and local classical strategies $\{p_A(a|x,\lambda)\}$ and $\{p_B(b|y,\lambda)\}$. On the contrary, if there exist measurements for which the resulting distribution $\{p(ab|xy)\}$ does not admit a LHV model, then the state $\rho_{AB}$ is deemed \textit{Bell nonlocal}. While any Bell-nonlocal state must be entangled, there exist entangled states that are Bell local~\cite{werner1989quantum,Barrett_2002}.

To demonstrate the nonlocality of a given distribution $\{p(ab|xy)\}$, we use Bell inequalities. These are linear functionals of the probabilities of the form
\begin{align}
    \label{eq:Bell_ineq}
    S \coloneqq \sum_{abxy} B_{xy}^{ab} \, p(ab|xy) \leq S_l,
\end{align}
where $\{B_{xy}^{ab}\}$ are real coefficients, and $S_l$ denotes the local bound, i.e., the maximum value of the Bell parameter $S$ for any distribution admitting a LHV model. 

Remarkably, it has been shown that the nonlocality can be activated when the Bell test is performed on multiple copies of certain Bell-local states. Namely, Ref.~\cite{palazuelos2012superactivation} demonstrated the existence of states $\rho_{AB}$ that are Bell local, whereas $\rho_{AB}^{\otimes n}$ is Bell nonlocal for some number of copies $n$. This effect is termed many-copy activation of Bell nonlocality. 

\section{Catalytic activation of Bell nonlocality}

In this work, we propose and investigate a new form of activation of quantum Bell nonlocality, based on the effect of quantum catalysis. We consider the scenario depicted in Fig.~\ref{fig:1}(b), where Alice and Bob share a Bell-local state $\rho_{AB}$, as well as an entangled catalyst state $\omega_{C_AC_B}$. Each party acts jointly on their local subsystems (i.e. Alice acts jointly on subsystems $AC_A\mapsto A'C_A$ and Bob on $BC_B\mapsto B'C_B$), transforming their global state $\rho_{AB}\otimes\omega_{C_AC_B}\mapsto\tau_{A'B'C_AC_B}$. This transformation is termed catalytic if the catalyst is returned unchanged, namely, if the marginal state of the catalyst fulfills
\begin{align}
    \tr_{A'B'}(\tau_{A'B'C_A C_B})=\omega_{C_A C_B} \,.
\end{align}
Let us now focus on the final state of the system, given by the marginal state 
\begin{align}
    \tau_{A'B'}= \tr_{C_A C_B}(\tau_{A'B'C_A C_B}) \,.
\end{align} 
Interestingly, as we will see below, it is possible to obtain a final state that is Bell nonlocal. Hence, we call this procedure catalytic activation of Bell nonlocality. 

Importantly, the protocol is deterministic, and the parties perform only local operations, with no need for classical communication (or even shared randomness). Moreover, since this pre-processing of the state is performed before the Bell test, no loophole is opened. At the end of the transformation, the catalyst is classically correlated with the original system, i.e., $\tau_{A'B'C_AC_B} \neq \tau_{A'B'} \otimes \omega_{C_A C_B}$. 

Below, we show that such a catalytic activation of Bell nonlocality is possible. In particular, we show that when a state can be activated with many copies, it can also be catalytically activated.
Moreover, we show that catalytic activation of Bell nonlocality is possible using a catalyst that is itself prepared in a Bell-local state. 

\begin{theorem}\label{thm1}
    Let $\rho_{AB}$ be a Bell-local entangled state, such that $n$ copies exhibit Bell nonlocality, i.e., such that $\rho_{AB}^{\otimes n}$ violates a Bell inequality for some finite $n$. Then, one copy of $\rho_{AB}$ exhibits catalytic Bell nonlocality activation, with the state of the catalyst being given by 
    \begin{align}\label{eq: catalyst}
        \omega_{C_A C_B} := \frac{1}{n}  \sum_{i = 0}^{n-1} \rho_{AB}^{\otimes i} \ot \sigma_{AB}^{\otimes (n-1-i)} \ot[ii]_{\widetilde R_A \widetilde R_B} \, ,
    \end{align}
    where $[ii]_{\widetilde R_A \widetilde R_B} \equiv \dyad{i}_{\widetilde R_A} \otimes \dyad{i}_{\widetilde R_B}$ denotes the state of a shared classical register and $\sigma_{AB}=\sigma_A \otimes \sigma_B$ is an arbitrary product state.
    
    In other words, there exist catalytic local transformations for Alice and Bob that deterministically map $\rho_{AB} \otimes \omega_{C_A C_B}$ to a final state $\tau_{A'B'C_AC_B}$ such that the marginal state of the system $\tau_{A'B'}= \tr_{C_A C_B}(\tau_{A'B'C_A C_B})$ is Bell nonlocal, while the marginal state of the catalyst satisfies the catalytic condition $\tr_{A'B'}(\tau_{A'B'C_A C_B})=\omega_{C_A C_B}$. 
\end{theorem}

We now present the proof of Theorem~\ref{thm1}, which works in two steps. First, we discuss the transformation of the initial state $\rho_{AB} \otimes \omega_{C_A C_B}$ via local operations. Specifically, we prove the following result. 

\begin{lemma} \label{lem: cata}
    By means of local operations any state $\rho_{AB}$ can be catalytically transformed to 
    \begin{equation} \label{eq: out catalyst n}
        \tau_{A' B'} =  \frac{1}{n}\, \rho_{AB}^{\otimes n} \otimes [00]_{R_A R_B} + \frac{n-1}{n}\, \sigma_{AB}^{\otimes n} \otimes [11]_{R_A R_B}
    \end{equation}
    where $A' \equiv A_1\dots A_n R_A$ and $B' \equiv B_1\dots B_n R_B$ are composed of $n\in \mathbb{N}$ copies of the system $A$ and $B$ and classical bits $R_A$ and $R_B$, and where $\sigma_{AB}=\sigma_A \otimes \sigma_B$ is an arbitrary product state. The state of the catalyst is given by \eqref{eq: catalyst}.
\end{lemma}

The proof of Lemma~\ref{lem: cata} is detailed below. It follows the main idea of a construction in Ref.~\cite{duan2005multiple}, adapted here to the context of Bell nonlocality. Ref.~\cite{duan2005multiple} investigated the connection between multicopy and catalytic transformations in entanglement theory (see also, e.g., Refs.~\cite{PhysRevLett.127.080502, char2023catalytic,Ganardi_2024, wilming2021entropy}). Our construction is in fact also relevant in this context. Notice that $\tau_{A' B'}$ cannot be prepared from $\rho_{AB}$ by means of stochastic local operations and classical communication in general (take, e.g., $\rho_{AB}$ to be maximally entangled). This shows that the addition of catalysts significantly boosts the power of local state transformations. These questions will be further explored in a forthcoming paper~\cite{inprep}. 

\begin{proof}
For clarity, we present here the proof for the simplest case of $n=2$, meaning that $\rho_{AB}$ is Bell local and $\rho_{AB}^{\otimes 2}$ is Bell nonlocal; the general proof is in Appendix~\ref{app:proof}. In the following, the collective catalyst systems are denoted as $C_AC_B$, with individual subsystems denoted by tilde variables, i.e., $C_A\equiv\tilde{A}\tilde{R}_A$, where $\tilde{A}$ is a quantum system and $\tilde{R}_A$ is a classical register; and equivalently for subscript $B$. The initial collective target systems are denoted as $AB$. The collective output target systems, after the local transformation, are denoted as $A'B'$, with individual subsystems denoted as $A'\equiv A_1 A_2 R_A$, where $A_1$ and $A_2$ are copies of quantum subsystem $A$, and $R_A$ is a classical register; and equivalently for subscript $B$.

The initial state of the system and catalyst is given by
\begin{align}
\begin{split}\label{eq:initialstate}
    \rho_{A B} \otimes \omega_{C_AC_B} = \frac{1}{2}\big(& \rho_{A B}\otimes\sigma_{\widetilde{A}\widetilde{B}} \otimes [00]_{\widetilde{R}_A \widetilde{R}_B} \\
    &+ \rho_{A B}\otimes\rho_{\widetilde{A}\widetilde{B}} \otimes [11]_{\widetilde{R}_A \widetilde{R}_B} \big).
\end{split}
\end{align}
Alice and Bob then implement local transformations that map their received local systems $A C_A \mapsto A' C_A$ and $B C_B \mapsto B' C_B$. Hence, their local operations enlarge their system dimensions. By measuring their copy of the classical bits $\widetilde{R}_A\widetilde{R}_B$, both parties learn in which branch of the mixture in Eq.~\eqref{eq:initialstate} they are. 
If the bit value is $0$, they (i) swap systems $A$ with $\widetilde{A}$, and $B$ with $\widetilde{B}$, preparing the catalyst on state $\rho_{\widetilde{A}\widetilde{B}}$ with classical register flipped to $[11]_{\widetilde{R}_A\widetilde{R}_B}$; and (ii) locally prepare output system $A'B'$ in the product state $\sigma_{A_1B_1}\otimes\sigma_{A_2B_2}$ and set the new classical register to $[11]_{R_A R_B}$. 
On the other hand, if the bit's value is $1$, they (i) prepare the output systems $A'B'$ in the state $\rho_{A_1B_1}\ot\rho_{A_2B_2}$, using $\rho_{AB}\otimes\rho_{\widetilde{A}\widetilde{B}}$, and prepare the new classical register on value $[00]_{R_AR_B}$; and (ii) locally prepare a new catalyst system in the product state $\sigma_{\widetilde{A}\widetilde{B}}$ with classical register flipped to value $[00]_{\widetilde{R}_A\widetilde{R}_B}$.

The final state of their joint system is given by
\begin{align}
\begin{split}
     \tau_{A'B'C_AC_B} &= \\ \frac{1}{2}\big(\rho_{A_1B_1}&\otimes\rho_{A_2B_2} \otimes [00]_{R_AR_B} \otimes  \sigma_{\widetilde{A}
     \widetilde{B}} \otimes [00]_{\widetilde{R}_A \widetilde{R}_B} \\
     + \, \sigma_{A_1B_1}&\otimes\sigma_{A_2B_2} \otimes [11]_{R_AR_B} \otimes \rho_{\widetilde{A}
     \widetilde{B}} \otimes [11]_{\widetilde{R}_A \widetilde{R}_B}\big).
\end{split}
\end{align}
We can check that the state of the catalyst, given by the marginal state $\tau_{C_A C_B}\coloneqq\tr_{A'B'}(\tau_{A'B'C_A C_B})=\omega_{C_A C_B}$, remains unchanged. 
This ensures that the transformation is indeed catalytic.
The catalyst has therefore \textit{not} been consumed, and can in principle be used again in a future transformation.

At the same time, the system is now described by the marginal state $\tau_{A'B'}\coloneqq\tr_{C_A C_B}(\tau_{A'B'C_A C_B})$ given by
\begin{align}\label{eq: two copy cata}
    \tau_{A'B'} = \frac{1}{2} \rho_{AB}^{\ot 2} \ot [00]_{R_{A} R_B} + \frac{1}{2} \sigma_{AB}^{\ot 2} \ot [11]_{R_A R_B} \, ,
\end{align}
which corresponds to the state in \eqref{eq: out catalyst n} for $n=2$. 
\end{proof}

Now, to complete the proof of the theorem, we must show that the final state of the system, given by \eqref{eq: out catalyst n}, is Bell nonlocal. To do so, we use the following result.

\begin{lemma}\label{lem: bell}
    Consider a state $\xi_{AB}$ that violates a given Bell inequality. Then the state
    \begin{equation}
      \tau = p \, \xi_{AB}\ot [00]_{R_A R_B} + (1-p) \, \sigma_{AB}\ot [11]_{R_A R_B},
    \end{equation}
   also violates the same Bell inequality, where $p>0$ and $\sigma_{AB}$ is an arbitrary quantum state.
\end{lemma}
\begin{proof}
    By assumption there exists a Bell inequality, of the form of Eq.~\eqref{eq:Bell_ineq} violated by $\xi_{AB}$. That is, there exist local measurements for the parties yielding the Bell score $S_l +\Delta >S_l$. To prove nonlocality of the state $\tau$, consider the following local measurements. The parties first read out the correlated states of the classical registers $[ii]_{R_AR_B}$. For $i=0$ they know that the systems $AB$ are in the state $\rho_{AB}$ and perform the measurements leading to the score $S_l +\Delta$. For $i=1$ they perform the local deterministic strategy that saturates the local bound $S_l$. Hence the expected Bell parameter is 
    \begin{equation}\label{eq::score}
        S = S_l + p \, \Delta > S_l,
    \end{equation}
    exceeding the local bound. Note that this proposition can be straightforwardly generalized to multipartite states.
\end{proof}

Setting $\xi_{AB} = \rho_{AB}$, we get from the above proposition that the final state of the system, given by Eq.~\eqref{eq: out catalyst n}, indeed violates a Bell inequality, which completes the proof of Theorem~\ref{thm1}.

More generally, Theorem~\ref{thm1} establishes a clear connection between catalytic activation of Bell nonlocality and many-copy activation of nonlocality~\cite{palazuelos2012superactivation}. We can now exploit further this connection to identify special classes of entangled states for which catalytic Bell nonlocality occurs. In particular, many-copy activation of Bell nonlocality has been shown to be possible for any state $\rho_{AB}$ with singlet fraction $F(\rho_{AB}) > 1/d$, where $d$ is the local dimension ~\cite{cavalcanti2013all}. The singlet fraction $F(\rho_{AB})$~\cite{PhysRevA.60.1888} is the largest overlap of the state $\rho_{AB}$ with a maximally entangled state, as in  
\begin{equation}\label{ent_fra}
   F(\rho_{AB}) \coloneqq \max_{\phi_{AB}}  \bra{\phi_{AB}}\rho_{AB} \ket{\phi_{AB}},
\end{equation}
where the optimization is given over all maximally-entangled states $\ket{\phi_{AB}}$. Hence, from Theorem~\ref{thm1}, we obtain the following result.

\begin{corollary}\label{corrol: singlet fraction}
    Every state $\rho_{AB}$ with local dimension $d$ and entanglement fraction $F(\rho_{AB}) > 1/d$ exhibits catalytic Bell nonlocality. In particular, this includes states $\rho_{AB}$ that are Bell local for arbitrary local measurements. An example is the two-qubit isotropic state
    \begin{align}
        \rho_{AB} = V \ket{\phi^+}\bra{\phi^+} + (1-V) \openone/4
    \end{align}
    for visibilities $1/3 < V \leq 1/2$ (corresponding to singlet fraction $1/2 < F(\rho_{AB}) \leq 5/8$~\cite{werner1989quantum,zhang24exact,renner24compatibility}.
\end{corollary}

\section{Activation with a Bell-local catalyst}

An intriguing feature of our protocol is that, in certain cases, Bell nonlocality can be catalytically activated using a catalyst that is itself prepared in a Bell-local state. In this sense, we can have catalytic superactivation of Bell nonlocality, since we start from two states $\rho_{AB}$ and $\omega_{C_A C_B}$ that are both Bell local. 

To see this, consider the following argument. Given a Bell-local state $\rho_{AB}$, take the smallest number of copies $n$ such that many-copy activation Bell nonlocality is possible. That is, such that $\rho_{AB}^{\otimes n}$ violates a Bell inequality but $\rho_{AB}^{\otimes m}$ with $m<n$ does not. Applying Theorem~\ref{thm1} to this case, we obtain catalytic activation of Bell nonlocality for $\rho_{AB}$ using a catalyst in the state in Eq.~\eqref{eq: catalyst} which is also Bell local; this is because the catalyst is a mixture containing terms of the form $\rho_{AB}^{\otimes m}$ with $m<n$, which are all Bell local by hypothesis, and terms with product states.

This effect is particularly interesting from the perspective of quantum catalysis. To the best of our knowledge, all known examples of catalysis in quantum information, where one seeks to activate a certain resource catalytically, require that the catalyst already possesses the resource in order to activate it in the main system. In contrast, we have shown that Bell nonlocality (the resource here) can be activated using a catalyst that is useless (here Bell local).

\section{Catalytic activation of CHSH}

We can also consider the problem of activating the nonlocality of a quantum state for a specific Bell inequality. Here we discuss the Clauser-Horne-Shimony-Holt Bell inequality \cite{chsh}, which is the most commonly used test of nonlocality and central to many device-independent protocols for randomness certification~\cite{Pironio2010,colbeck2011} and quantum key distribution~\cite{Acin2007,ArnonFriedman2018,Nadlinger2022}.

Catalytic activation of CHSH violation can be demonstrated applying our main theorem to the two-copy activation result of Ref.~\cite{Navascues2011}. This work presents states $\rho_{AB}$ (of local dimensions $d=8$) which do not violate the CHSH Bell inequality (yet it is not clear if this state is Bell local or not), while two copies $\rho_{AB}^{\ot 2}$ lead to violation of CHSH. Note that in the catalytic activation process, the catalyst contains a single copy of $\rho_{AB}$, hence this state cannot violate CHSH.

\section{Other forms of catalytic activation of Bell nonlocality}

So far, we have presented a protocol for catalytic activation of Bell nonlocality, which consists of a local catalytic transformation of the entangled state followed by the Bell test. Alternatively, one could consider a scenario where the catalyst must only be returned after the local measurements of Alice and Bob, as depicted in Fig.~\ref{fig:1}(c). This more general scenario could, in principle, provide more possibilities for catalytically activating Bell nonlocality.

In this case, the local measurements of Alice and Bob must be described by quantum instruments 
$\{\mathcal{I}^{a|x}_{AC_A}\}$ and $\{\mathcal{I}^{b|y}_{BC_B}\}$, returning not only a classical outcome, but also a quantum system, from which the catalyst should eventually be recovered. For a given pair of $x$ and $y$, the instruments output the global classical-quantum state
\begin{align}\label{eq: cat variant 2}
    \tau_{O_A O_B C_A C_B}^{(x,y)} &=
    \sum_{a,b} p(ab|xy)\,\,  [ab]_{O_A O_B} \otimes \omega_{C_AC_B}^{(a,b,x,y)},
\end{align}
where $ p(ab|xy)\, \omega_{C_AC_B}^{(a,b,x,y)} =(\mathcal{I}^{a|x}_{AC_A} \ot \mathcal{I}^{b|y}_{BC_B})[\rho_{AB} \ot \omega_{C_AC_B}]$, and $O_A$ with $O_B$ are registers storing the classical outputs $a$ and $b$. Starting from a Bell-local state $\rho_{AB}$ we then have catalytic activation of nonlocality if the produced correlations $p(a,b|x,y)$ are nonlocal, while the state of the catalyst is preserved. There are now several variants on how to formalize this requirement, depending on whether the marginal state of the catalyst must be unchanged for all classical inputs and outputs, only for all classical inputs, or simply on average. In App.~\ref{app: intrumental} we discuss all these variants in detail, and compare them to catalytic activation of nonlocality via state transformation as discussed in the previous sections. There we also comment on the role of shared randomness. 

Let us also note that Ref.~\cite{Karvonen2021} considered yet a different approach to catalysis in the context of Bell nonlocality, using local wirings of correlations, rather than local operations on quantum states as we do here. They showed that in this formulation, catalytic activation of nonlocality is not possible.

\section{Discussion and Outlook}

We have uncovered a novel mechanism for activating Bell nonlocality based on quantum catalysis. By combining a Bell-local entangled state with a catalyst, we demonstrate that Bell nonlocality can be activated, while the state of the catalyst remains unchanged. 

A striking aspect of this process is that it does not require the catalyst to be prepared in a Bell-nonlocal state. Indeed, we showed that using a Bell-local state as a catalyst can lead to catalytic activation of Bell nonlocality. This point is also relevant from the perspective of quantum resource theories, where catalysis is usually possible only when the catalyst already possesses the resource to be activated in the system. 

A relevant open question is to understand for which states can Bell nonlocality be catalytically activated, and in particular, if this would be possible for every entangled state. Here we have identified large classes of entangled states for which this is possible, drawing on a connection to many-copy activation. As a first step, it would be interesting to find further examples of catalytic activation of Bell nonlocality that do not rely on many-copy activation. Another relevant question is to understand how catalytic activation of Bell nonlocality relates to other forms of nonlocality activation. At this point, we can see that this effect is different from hidden nonlocality. Indeed, we have shown that every entangled isotropic state can be catalytically activated, while some of these state remain Bell local even after local filtering~\cite{Hirsch2016}.

Finally, it is worth noting that some of our results, in particular Lemma~\ref{lem: cata}, are of independent interest and can be used to demonstrate catalytic activation in different contexts. In particular, any property of a quantum state that is absent in a single copy, but that can be activated with many copies, could also be activated catalytically according to Lemma~\ref{lem: cata}, as long $\tau_{A'B'}$ in Eq.~\eqref{eq: out catalyst n} also carries this property. For example, it has been shown that metrological advantage of certain entangled states can be activated in the many-copy regime~\cite{Toth2020,Trenyi2024}. From our results, it follows that catalytic activation of metrological advantage is also possible. It would also be interesting to investigate possible applications for device-independent protocols.
\\

\noindent\textit{Acknowledgments.---}
We thank Marco T\'ulio Quintino, Tam\'as V\'ertesi, and Paul Skrzypczyk for discussions. The authors acknowledge funding from the Swiss National Science Foundation (SNSF) through NCCR SwissMAP (project~182902), project~192244, and the Swiss Postdoctoral Fellowship (project~216979). P.L.-B. also acknowledges funding from Polish National Agency for Academic Exchange (NAWA) through grant BPN/PPO/2023/1/00018/U/00001.

\bibliography{catalysis}

\begin{thebibliography}{53}%
\makeatletter
\providecommand \@ifxundefined [1]{%
 \@ifx{#1\undefined}
}%
\providecommand \@ifnum [1]{%
 \ifnum #1\expandafter \@firstoftwo
 \else \expandafter \@secondoftwo
 \fi
}%
\providecommand \@ifx [1]{%
 \ifx #1\expandafter \@firstoftwo
 \else \expandafter \@secondoftwo
 \fi
}%
\providecommand \natexlab [1]{#1}%
\providecommand \enquote  [1]{``#1''}%
\providecommand \bibnamefont  [1]{#1}%
\providecommand \bibfnamefont [1]{#1}%
\providecommand \citenamefont [1]{#1}%
\providecommand \href@noop [0]{\@secondoftwo}%
\providecommand \href [0]{\begingroup \@sanitize@url \@href}%
\providecommand \@href[1]{\@@startlink{#1}\@@href}%
\providecommand \@@href[1]{\endgroup#1\@@endlink}%
\providecommand \@sanitize@url [0]{\catcode `\\12\catcode `\$12\catcode
  `\&12\catcode `\#12\catcode `\^12\catcode `\_12\catcode `\%12\relax}%
\providecommand \@@startlink[1]{}%
\providecommand \@@endlink[0]{}%
\providecommand \url  [0]{\begingroup\@sanitize@url \@url }%
\providecommand \@url [1]{\endgroup\@href {#1}{\urlprefix }}%
\providecommand \urlprefix  [0]{URL }%
\providecommand \Eprint [0]{\href }%
\providecommand \doibase [0]{https://doi.org/}%
\providecommand \selectlanguage [0]{\@gobble}%
\providecommand \bibinfo  [0]{\@secondoftwo}%
\providecommand \bibfield  [0]{\@secondoftwo}%
\providecommand \translation [1]{[#1]}%
\providecommand \BibitemOpen [0]{}%
\providecommand \bibitemStop [0]{}%
\providecommand \bibitemNoStop [0]{.\EOS\space}%
\providecommand \EOS [0]{\spacefactor3000\relax}%
\providecommand \BibitemShut  [1]{\csname bibitem#1\endcsname}%
\let\auto@bib@innerbib\@empty
\bibitem [{\citenamefont {Bell}(1964)}]{bell1964einstein}%
  \BibitemOpen
  \bibfield  {author} {\bibinfo {author} {\bibfnamefont {J.~S.}\ \bibnamefont
  {Bell}},\ }\bibfield  {title} {\bibinfo {title} {{On the Einstein Podolsky
  Rosen paradox}},\ }\href
  {https://doi.org/10.1103/PhysicsPhysiqueFizika.1.195} {\bibfield  {journal}
  {\bibinfo  {journal} {Physics Physique Fizika}\ }\textbf {\bibinfo {volume}
  {1}},\ \bibinfo {pages} {195} (\bibinfo {year} {1964})}\BibitemShut {NoStop}%
\bibitem [{\citenamefont {Brunner}\ \emph {et~al.}(2014)\citenamefont
  {Brunner}, \citenamefont {Cavalcanti}, \citenamefont {Pironio}, \citenamefont
  {Scarani},\ and\ \citenamefont {Wehner}}]{Brunner_2014}%
  \BibitemOpen
  \bibfield  {author} {\bibinfo {author} {\bibfnamefont {N.}~\bibnamefont
  {Brunner}}, \bibinfo {author} {\bibfnamefont {D.}~\bibnamefont {Cavalcanti}},
  \bibinfo {author} {\bibfnamefont {S.}~\bibnamefont {Pironio}}, \bibinfo
  {author} {\bibfnamefont {V.}~\bibnamefont {Scarani}},\ and\ \bibinfo {author}
  {\bibfnamefont {S.}~\bibnamefont {Wehner}},\ }\bibfield  {title} {\bibinfo
  {title} {Bell nonlocality},\ }\href
  {https://doi.org/10.1103/revmodphys.86.419} {\bibfield  {journal} {\bibinfo
  {journal} {Reviews of Modern Physics}\ }\textbf {\bibinfo {volume} {86}},\
  \bibinfo {pages} {419–478} (\bibinfo {year} {2014})},\ \Eprint
  {https://arxiv.org/abs/1303.2849} {arXiv:1303.2849 [quant-ph]} \BibitemShut
  {NoStop}%
\bibitem [{\citenamefont {Ac\'{\i}n}\ \emph {et~al.}(2007)\citenamefont
  {Ac\'{\i}n}, \citenamefont {Brunner}, \citenamefont {Gisin}, \citenamefont
  {Massar}, \citenamefont {Pironio},\ and\ \citenamefont {Scarani}}]{Acin2007}%
  \BibitemOpen
  \bibfield  {author} {\bibinfo {author} {\bibfnamefont {A.}~\bibnamefont
  {Ac\'{\i}n}}, \bibinfo {author} {\bibfnamefont {N.}~\bibnamefont {Brunner}},
  \bibinfo {author} {\bibfnamefont {N.}~\bibnamefont {Gisin}}, \bibinfo
  {author} {\bibfnamefont {S.}~\bibnamefont {Massar}}, \bibinfo {author}
  {\bibfnamefont {S.}~\bibnamefont {Pironio}},\ and\ \bibinfo {author}
  {\bibfnamefont {V.}~\bibnamefont {Scarani}},\ }\bibfield  {title} {\bibinfo
  {title} {Device-independent security of quantum cryptography against
  collective attacks},\ }\href {https://doi.org/10.1103/PhysRevLett.98.230501}
  {\bibfield  {journal} {\bibinfo  {journal} {Phys. Rev. Lett.}\ }\textbf
  {\bibinfo {volume} {98}},\ \bibinfo {pages} {230501} (\bibinfo {year}
  {2007})},\ \Eprint {https://arxiv.org/abs/quant-ph/0702152}
  {arXiv:quant-ph/0702152} \BibitemShut {NoStop}%
\bibitem [{\citenamefont {Pironio}\ \emph {et~al.}(2010)\citenamefont
  {Pironio}, \citenamefont {Acín}, \citenamefont {Massar}, \citenamefont
  {de~la Giroday}, \citenamefont {Matsukevich}, \citenamefont {Maunz},
  \citenamefont {Olmschenk}, \citenamefont {Hayes}, \citenamefont {Luo},
  \citenamefont {Manning},\ and\ \citenamefont {Monroe}}]{Pironio2010}%
  \BibitemOpen
  \bibfield  {author} {\bibinfo {author} {\bibfnamefont {S.}~\bibnamefont
  {Pironio}}, \bibinfo {author} {\bibfnamefont {A.}~\bibnamefont {Acín}},
  \bibinfo {author} {\bibfnamefont {S.}~\bibnamefont {Massar}}, \bibinfo
  {author} {\bibfnamefont {A.~B.}\ \bibnamefont {de~la Giroday}}, \bibinfo
  {author} {\bibfnamefont {D.~N.}\ \bibnamefont {Matsukevich}}, \bibinfo
  {author} {\bibfnamefont {P.}~\bibnamefont {Maunz}}, \bibinfo {author}
  {\bibfnamefont {S.}~\bibnamefont {Olmschenk}}, \bibinfo {author}
  {\bibfnamefont {D.}~\bibnamefont {Hayes}}, \bibinfo {author} {\bibfnamefont
  {L.}~\bibnamefont {Luo}}, \bibinfo {author} {\bibfnamefont {T.~A.}\
  \bibnamefont {Manning}},\ and\ \bibinfo {author} {\bibfnamefont
  {C.}~\bibnamefont {Monroe}},\ }\bibfield  {title} {\bibinfo {title} {{Random
  numbers certified by Bell’s theorem}},\ }\href
  {https://doi.org/10.1038/nature09008} {\bibfield  {journal} {\bibinfo
  {journal} {Nature}\ }\textbf {\bibinfo {volume} {464}},\ \bibinfo {pages}
  {1021–1024} (\bibinfo {year} {2010})},\ \Eprint
  {https://arxiv.org/abs/0911.3427} {arXiv:0911.3427 [quant-ph]} \BibitemShut
  {NoStop}%
\bibitem [{\citenamefont {Colbeck}(2011)}]{colbeck2011}%
  \BibitemOpen
  \bibfield  {author} {\bibinfo {author} {\bibfnamefont {R.}~\bibnamefont
  {Colbeck}},\ }\href@noop {} {\bibinfo {title} {Quantum and relativistic
  protocols for secure multi-party computation}} (\bibinfo {year} {2011}),\
  \Eprint {https://arxiv.org/abs/0911.3814} {arXiv:0911.3814 [quant-ph]}
  \BibitemShut {NoStop}%
\bibitem [{\citenamefont {{\v{S}}upi{\'{c}}}\ and\ \citenamefont
  {Bowles}(2020)}]{supic2020self}%
  \BibitemOpen
  \bibfield  {author} {\bibinfo {author} {\bibfnamefont {I.}~\bibnamefont
  {{\v{S}}upi{\'{c}}}}\ and\ \bibinfo {author} {\bibfnamefont {J.}~\bibnamefont
  {Bowles}},\ }\bibfield  {title} {\bibinfo {title} {Self-testing of quantum
  systems: a review},\ }\href {https://doi.org/10.22331/q-2020-09-30-337}
  {\bibfield  {journal} {\bibinfo  {journal} {{Quantum}}\ }\textbf {\bibinfo
  {volume} {4}},\ \bibinfo {pages} {337} (\bibinfo {year} {2020})},\ \Eprint
  {https://arxiv.org/abs/1904.10042} {arXiv:1904.10042 [quant-ph]} \BibitemShut
  {NoStop}%
\bibitem [{\citenamefont {Gisin}(1991)}]{GISIN1991}%
  \BibitemOpen
  \bibfield  {author} {\bibinfo {author} {\bibfnamefont {N.}~\bibnamefont
  {Gisin}},\ }\bibfield  {title} {\bibinfo {title} {Bell's inequality holds for
  all non-product states},\ }\href
  {https://doi.org/https://doi.org/10.1016/0375-9601(91)90805-I} {\bibfield
  {journal} {\bibinfo  {journal} {Physics Letters A}\ }\textbf {\bibinfo
  {volume} {154}},\ \bibinfo {pages} {201} (\bibinfo {year}
  {1991})}\BibitemShut {NoStop}%
\bibitem [{\citenamefont {Werner}(1989)}]{werner1989quantum}%
  \BibitemOpen
  \bibfield  {author} {\bibinfo {author} {\bibfnamefont {R.~F.}\ \bibnamefont
  {Werner}},\ }\bibfield  {title} {\bibinfo {title} {{Quantum states with
  Einstein-Podolsky-Rosen correlations admitting a hidden-variable model}},\
  }\href {https://doi.org/10.1103/PhysRevA.40.4277} {\bibfield  {journal}
  {\bibinfo  {journal} {Phys. Rev. A}\ }\textbf {\bibinfo {volume} {40}},\
  \bibinfo {pages} {4277} (\bibinfo {year} {1989})}\BibitemShut {NoStop}%
\bibitem [{\citenamefont {Barrett}(2002)}]{Barrett_2002}%
  \BibitemOpen
  \bibfield  {author} {\bibinfo {author} {\bibfnamefont {J.}~\bibnamefont
  {Barrett}},\ }\bibfield  {title} {\bibinfo {title} {{Nonsequential
  positive-operator-valued measurements on entangled mixed states do not always
  violate a Bell inequality}},\ }\href
  {https://doi.org/10.1103/PhysRevA.65.042302} {\bibfield  {journal} {\bibinfo
  {journal} {Phys. Rev. A}\ }\textbf {\bibinfo {volume} {65}},\ \bibinfo
  {pages} {042302} (\bibinfo {year} {2002})},\ \Eprint
  {https://arxiv.org/abs/quant-ph/0107045} {arXiv:quant-ph/0107045}
  \BibitemShut {NoStop}%
\bibitem [{\citenamefont {Ac\'{\i}n}\ \emph {et~al.}(2006)\citenamefont
  {Ac\'{\i}n}, \citenamefont {Gisin},\ and\ \citenamefont
  {Toner}}]{acin2006grothendieck}%
  \BibitemOpen
  \bibfield  {author} {\bibinfo {author} {\bibfnamefont {A.}~\bibnamefont
  {Ac\'{\i}n}}, \bibinfo {author} {\bibfnamefont {N.}~\bibnamefont {Gisin}},\
  and\ \bibinfo {author} {\bibfnamefont {B.}~\bibnamefont {Toner}},\ }\bibfield
   {title} {\bibinfo {title} {Grothendieck's constant and local models for
  noisy entangled quantum states},\ }\href
  {https://doi.org/10.1103/PhysRevA.73.062105} {\bibfield  {journal} {\bibinfo
  {journal} {Phys. Rev. A}\ }\textbf {\bibinfo {volume} {73}},\ \bibinfo
  {pages} {062105} (\bibinfo {year} {2006})},\ \Eprint
  {https://arxiv.org/abs/quant-ph/0606138} {arXiv:quant-ph/0606138}
  \BibitemShut {NoStop}%
\bibitem [{\citenamefont {Augusiak}\ \emph {et~al.}(2015)\citenamefont
  {Augusiak}, \citenamefont {Demianowicz}, \citenamefont {Tura},\ and\
  \citenamefont {Ac\'{\i}n}}]{augusiak2015entanglement}%
  \BibitemOpen
  \bibfield  {author} {\bibinfo {author} {\bibfnamefont {R.}~\bibnamefont
  {Augusiak}}, \bibinfo {author} {\bibfnamefont {M.}~\bibnamefont
  {Demianowicz}}, \bibinfo {author} {\bibfnamefont {J.}~\bibnamefont {Tura}},\
  and\ \bibinfo {author} {\bibfnamefont {A.}~\bibnamefont {Ac\'{\i}n}},\
  }\bibfield  {title} {\bibinfo {title} {Entanglement and nonlocality are
  inequivalent for any number of parties},\ }\href
  {https://doi.org/10.1103/PhysRevLett.115.030404} {\bibfield  {journal}
  {\bibinfo  {journal} {Phys. Rev. Lett.}\ }\textbf {\bibinfo {volume} {115}},\
  \bibinfo {pages} {030404} (\bibinfo {year} {2015})},\ \Eprint
  {https://arxiv.org/abs/1407.3114} {arXiv:1407.3114 [quant-ph]} \BibitemShut
  {NoStop}%
\bibitem [{\citenamefont {Bowles}\ \emph {et~al.}(2016)\citenamefont {Bowles},
  \citenamefont {Francfort}, \citenamefont {Fillettaz}, \citenamefont
  {Hirsch},\ and\ \citenamefont {Brunner}}]{Bowles2016}%
  \BibitemOpen
  \bibfield  {author} {\bibinfo {author} {\bibfnamefont {J.}~\bibnamefont
  {Bowles}}, \bibinfo {author} {\bibfnamefont {J.}~\bibnamefont {Francfort}},
  \bibinfo {author} {\bibfnamefont {M.}~\bibnamefont {Fillettaz}}, \bibinfo
  {author} {\bibfnamefont {F.}~\bibnamefont {Hirsch}},\ and\ \bibinfo {author}
  {\bibfnamefont {N.}~\bibnamefont {Brunner}},\ }\bibfield  {title} {\bibinfo
  {title} {Genuinely multipartite entangled quantum states with fully local
  hidden variable models and hidden multipartite nonlocality},\ }\href
  {https://doi.org/10.1103/PhysRevLett.116.130401} {\bibfield  {journal}
  {\bibinfo  {journal} {Phys. Rev. Lett.}\ }\textbf {\bibinfo {volume} {116}},\
  \bibinfo {pages} {130401} (\bibinfo {year} {2016})},\ \Eprint
  {https://arxiv.org/abs/1511.08401} {arXiv:1511.08401 [quant-ph]} \BibitemShut
  {NoStop}%
\bibitem [{\citenamefont {Popescu}(1995)}]{popescu1995bell}%
  \BibitemOpen
  \bibfield  {author} {\bibinfo {author} {\bibfnamefont {S.}~\bibnamefont
  {Popescu}},\ }\bibfield  {title} {\bibinfo {title} {{Bell's Inequalities and
  Density Matrices: Revealing ``Hidden'' Nonlocality}},\ }\href
  {https://doi.org/10.1103/PhysRevLett.74.2619} {\bibfield  {journal} {\bibinfo
   {journal} {Phys. Rev. Lett.}\ }\textbf {\bibinfo {volume} {74}},\ \bibinfo
  {pages} {2619} (\bibinfo {year} {1995})},\ \Eprint
  {https://arxiv.org/abs/0908.1583} {arXiv:0908.1583 [quant-ph/9502005]}
  \BibitemShut {NoStop}%
\bibitem [{\citenamefont {Gisin}(1996)}]{gisin1996hidden}%
  \BibitemOpen
  \bibfield  {author} {\bibinfo {author} {\bibfnamefont {N.}~\bibnamefont
  {Gisin}},\ }\bibfield  {title} {\bibinfo {title} {Hidden quantum nonlocality
  revealed by local filters},\ }\href
  {https://doi.org/10.1016/S0375-9601(96)80001-6} {\bibfield  {journal}
  {\bibinfo  {journal} {Physics Letters A}\ }\textbf {\bibinfo {volume}
  {210}},\ \bibinfo {pages} {151} (\bibinfo {year} {1996})}\BibitemShut
  {NoStop}%
\bibitem [{\citenamefont {Masanes}\ \emph {et~al.}(2008)\citenamefont
  {Masanes}, \citenamefont {Liang},\ and\ \citenamefont
  {Doherty}}]{Masanes2008}%
  \BibitemOpen
  \bibfield  {author} {\bibinfo {author} {\bibfnamefont {L.}~\bibnamefont
  {Masanes}}, \bibinfo {author} {\bibfnamefont {Y.-C.}\ \bibnamefont {Liang}},\
  and\ \bibinfo {author} {\bibfnamefont {A.~C.}\ \bibnamefont {Doherty}},\
  }\bibfield  {title} {\bibinfo {title} {All bipartite entangled states display
  some hidden nonlocality},\ }\href
  {https://doi.org/10.1103/PhysRevLett.100.090403} {\bibfield  {journal}
  {\bibinfo  {journal} {Phys. Rev. Lett.}\ }\textbf {\bibinfo {volume} {100}},\
  \bibinfo {pages} {090403} (\bibinfo {year} {2008})},\ \Eprint
  {https://arxiv.org/abs/quant-ph/0703268} {arXiv:quant-ph/0703268}
  \BibitemShut {NoStop}%
\bibitem [{\citenamefont {Hirsch}\ \emph {et~al.}(2013)\citenamefont {Hirsch},
  \citenamefont {Quintino}, \citenamefont {Bowles},\ and\ \citenamefont
  {Brunner}}]{hirsch2013genuine}%
  \BibitemOpen
  \bibfield  {author} {\bibinfo {author} {\bibfnamefont {F.}~\bibnamefont
  {Hirsch}}, \bibinfo {author} {\bibfnamefont {M.~T.}\ \bibnamefont
  {Quintino}}, \bibinfo {author} {\bibfnamefont {J.}~\bibnamefont {Bowles}},\
  and\ \bibinfo {author} {\bibfnamefont {N.}~\bibnamefont {Brunner}},\
  }\bibfield  {title} {\bibinfo {title} {Genuine hidden quantum nonlocality},\
  }\href {https://doi.org/10.1103/PhysRevLett.111.160402} {\bibfield  {journal}
  {\bibinfo  {journal} {Phys. Rev. Lett.}\ }\textbf {\bibinfo {volume} {111}},\
  \bibinfo {pages} {160402} (\bibinfo {year} {2013})},\ \Eprint
  {https://arxiv.org/abs/1307.4404} {arXiv:1307.4404 [quant-ph]} \BibitemShut
  {NoStop}%
\bibitem [{\citenamefont {Palazuelos}(2012)}]{palazuelos2012superactivation}%
  \BibitemOpen
  \bibfield  {author} {\bibinfo {author} {\bibfnamefont {C.}~\bibnamefont
  {Palazuelos}},\ }\bibfield  {title} {\bibinfo {title} {Superactivation of
  quantum nonlocality},\ }\href
  {https://doi.org/10.1103/PhysRevLett.109.190401} {\bibfield  {journal}
  {\bibinfo  {journal} {Phys. Rev. Lett.}\ }\textbf {\bibinfo {volume} {109}},\
  \bibinfo {pages} {190401} (\bibinfo {year} {2012})},\ \Eprint
  {https://arxiv.org/abs/1205.3118} {arXiv:1205.3118 [quant-ph]} \BibitemShut
  {NoStop}%
\bibitem [{\citenamefont {Cavalcanti}\ \emph {et~al.}(2013)\citenamefont
  {Cavalcanti}, \citenamefont {Ac\'{\i}n}, \citenamefont {Brunner},\ and\
  \citenamefont {V\'ertesi}}]{cavalcanti2013all}%
  \BibitemOpen
  \bibfield  {author} {\bibinfo {author} {\bibfnamefont {D.}~\bibnamefont
  {Cavalcanti}}, \bibinfo {author} {\bibfnamefont {A.}~\bibnamefont
  {Ac\'{\i}n}}, \bibinfo {author} {\bibfnamefont {N.}~\bibnamefont {Brunner}},\
  and\ \bibinfo {author} {\bibfnamefont {T.}~\bibnamefont {V\'ertesi}},\
  }\bibfield  {title} {\bibinfo {title} {All quantum states useful for
  teleportation are nonlocal resources},\ }\href
  {https://doi.org/10.1103/PhysRevA.87.042104} {\bibfield  {journal} {\bibinfo
  {journal} {Phys. Rev. A}\ }\textbf {\bibinfo {volume} {87}},\ \bibinfo
  {pages} {042104} (\bibinfo {year} {2013})},\ \Eprint
  {https://arxiv.org/abs/1207.5485} {arXiv:1207.5485 [quant-ph]} \BibitemShut
  {NoStop}%
\bibitem [{\citenamefont {Contreras-Tejada}\ \emph {et~al.}(2022)\citenamefont
  {Contreras-Tejada}, \citenamefont {Palazuelos},\ and\ \citenamefont
  {de~Vicente}}]{Contreras2022}%
  \BibitemOpen
  \bibfield  {author} {\bibinfo {author} {\bibfnamefont {P.}~\bibnamefont
  {Contreras-Tejada}}, \bibinfo {author} {\bibfnamefont {C.}~\bibnamefont
  {Palazuelos}},\ and\ \bibinfo {author} {\bibfnamefont {J.~I.}\ \bibnamefont
  {de~Vicente}},\ }\bibfield  {title} {\bibinfo {title} {Asymptotic survival of
  genuine multipartite entanglement in noisy quantum networks depends on the
  topology},\ }\href {https://doi.org/10.1103/PhysRevLett.128.220501}
  {\bibfield  {journal} {\bibinfo  {journal} {Phys. Rev. Lett.}\ }\textbf
  {\bibinfo {volume} {128}},\ \bibinfo {pages} {220501} (\bibinfo {year}
  {2022})},\ \Eprint {https://arxiv.org/abs/2106.04634} {arXiv:2106.04634
  [quant-ph]} \BibitemShut {NoStop}%
\bibitem [{\citenamefont {Sen(De)}\ \emph {et~al.}(2005)\citenamefont
  {Sen(De)}, \citenamefont {Sen}, \citenamefont {Brukner}, \citenamefont
  {Bu\ifmmode~\check{z}\else \v{z}\fi{}ek},\ and\ \citenamefont
  {\ifmmode~\dot{Z}\else \.{Z}\fi{}ukowski}}]{sen2005entanglement}%
  \BibitemOpen
  \bibfield  {author} {\bibinfo {author} {\bibfnamefont {A.}~\bibnamefont
  {Sen(De)}}, \bibinfo {author} {\bibfnamefont {U.}~\bibnamefont {Sen}},
  \bibinfo {author} {\bibfnamefont {C.}~\bibnamefont {Brukner}}, \bibinfo
  {author} {\bibfnamefont {V.}~\bibnamefont {Bu\ifmmode~\check{z}\else
  \v{z}\fi{}ek}},\ and\ \bibinfo {author} {\bibfnamefont {M.}~\bibnamefont
  {\ifmmode~\dot{Z}\else \.{Z}\fi{}ukowski}},\ }\bibfield  {title} {\bibinfo
  {title} {Entanglement swapping of noisy states: A kind of superadditivity in
  nonclassicality},\ }\href {https://doi.org/10.1103/PhysRevA.72.042310}
  {\bibfield  {journal} {\bibinfo  {journal} {Phys. Rev. A}\ }\textbf {\bibinfo
  {volume} {72}},\ \bibinfo {pages} {042310} (\bibinfo {year} {2005})},\
  \Eprint {https://arxiv.org/abs/quant-ph/0311194} {arXiv:quant-ph/0311194}
  \BibitemShut {NoStop}%
\bibitem [{\citenamefont {Cavalcanti}\ \emph {et~al.}(2011)\citenamefont
  {Cavalcanti}, \citenamefont {Almeida}, \citenamefont {Scarani},\ and\
  \citenamefont {Acín}}]{Cavalcanti2011}%
  \BibitemOpen
  \bibfield  {author} {\bibinfo {author} {\bibfnamefont {D.}~\bibnamefont
  {Cavalcanti}}, \bibinfo {author} {\bibfnamefont {M.~L.}\ \bibnamefont
  {Almeida}}, \bibinfo {author} {\bibfnamefont {V.}~\bibnamefont {Scarani}},\
  and\ \bibinfo {author} {\bibfnamefont {A.}~\bibnamefont {Acín}},\ }\bibfield
   {title} {\bibinfo {title} {Quantum networks reveal quantum nonlocality},\
  }\href {https://doi.org/10.1038/ncomms1193} {\bibfield  {journal} {\bibinfo
  {journal} {Nature Communications}\ }\textbf {\bibinfo {volume} {2}},\
  \bibinfo {pages} {184} (\bibinfo {year} {2011})},\ \Eprint
  {https://arxiv.org/abs/1010.0900} {arXiv:1010.0900 [quant-ph]} \BibitemShut
  {NoStop}%
\bibitem [{\citenamefont {Cavalcanti}\ \emph {et~al.}(2012)\citenamefont
  {Cavalcanti}, \citenamefont {Rabelo},\ and\ \citenamefont
  {Scarani}}]{cavalcanti2012nonlocalitytests}%
  \BibitemOpen
  \bibfield  {author} {\bibinfo {author} {\bibfnamefont {D.}~\bibnamefont
  {Cavalcanti}}, \bibinfo {author} {\bibfnamefont {R.}~\bibnamefont {Rabelo}},\
  and\ \bibinfo {author} {\bibfnamefont {V.}~\bibnamefont {Scarani}},\
  }\bibfield  {title} {\bibinfo {title} {Nonlocality tests enhanced by a third
  observer},\ }\href {https://doi.org/10.1103/PhysRevLett.108.040402}
  {\bibfield  {journal} {\bibinfo  {journal} {Phys. Rev. Lett.}\ }\textbf
  {\bibinfo {volume} {108}},\ \bibinfo {pages} {040402} (\bibinfo {year}
  {2012})},\ \Eprint {https://arxiv.org/abs/1110.3656} {arXiv:1110.3656
  [quant-ph]} \BibitemShut {NoStop}%
\bibitem [{\citenamefont {Bowles}\ \emph {et~al.}(2021)\citenamefont {Bowles},
  \citenamefont {Hirsch},\ and\ \citenamefont {Cavalcanti}}]{bowles2021single}%
  \BibitemOpen
  \bibfield  {author} {\bibinfo {author} {\bibfnamefont {J.}~\bibnamefont
  {Bowles}}, \bibinfo {author} {\bibfnamefont {F.}~\bibnamefont {Hirsch}},\
  and\ \bibinfo {author} {\bibfnamefont {D.}~\bibnamefont {Cavalcanti}},\
  }\bibfield  {title} {\bibinfo {title} {Single-copy activation of {B}ell
  nonlocality via broadcasting of quantum states},\ }\href
  {https://doi.org/10.22331/q-2021-07-13-499} {\bibfield  {journal} {\bibinfo
  {journal} {{Quantum}}\ }\textbf {\bibinfo {volume} {5}},\ \bibinfo {pages}
  {499} (\bibinfo {year} {2021})},\ \Eprint {https://arxiv.org/abs/2007.16034}
  {arXiv:2007.16034 [quant-ph]} \BibitemShut {NoStop}%
\bibitem [{\citenamefont {Boghiu}\ \emph {et~al.}(2023)\citenamefont {Boghiu},
  \citenamefont {Hirsch}, \citenamefont {Lin}, \citenamefont {Quintino},\ and\
  \citenamefont {Bowles}}]{Boghiu2023}%
  \BibitemOpen
  \bibfield  {author} {\bibinfo {author} {\bibfnamefont {E.-C.}\ \bibnamefont
  {Boghiu}}, \bibinfo {author} {\bibfnamefont {F.}~\bibnamefont {Hirsch}},
  \bibinfo {author} {\bibfnamefont {P.-S.}\ \bibnamefont {Lin}}, \bibinfo
  {author} {\bibfnamefont {M.~T.}\ \bibnamefont {Quintino}},\ and\ \bibinfo
  {author} {\bibfnamefont {J.}~\bibnamefont {Bowles}},\ }\bibfield  {title}
  {\bibinfo {title} {{Device-independent and semi-device-independent
  entanglement certification in broadcast Bell scenarios}},\ }\href
  {https://doi.org/10.21468/SciPostPhysCore.6.2.028} {\bibfield  {journal}
  {\bibinfo  {journal} {SciPost Phys. Core}\ }\textbf {\bibinfo {volume} {6}},\
  \bibinfo {pages} {028} (\bibinfo {year} {2023})},\ \Eprint
  {https://arxiv.org/abs/2111.06358} {arXiv:2111.06358 [quant-ph]} \BibitemShut
  {NoStop}%
\bibitem [{\citenamefont {Lipka-Bartosik}\ \emph {et~al.}(2024)\citenamefont
  {Lipka-Bartosik}, \citenamefont {Wilming},\ and\ \citenamefont
  {Ng}}]{review1}%
  \BibitemOpen
  \bibfield  {author} {\bibinfo {author} {\bibfnamefont {P.}~\bibnamefont
  {Lipka-Bartosik}}, \bibinfo {author} {\bibfnamefont {H.}~\bibnamefont
  {Wilming}},\ and\ \bibinfo {author} {\bibfnamefont {N.~H.~Y.}\ \bibnamefont
  {Ng}},\ }\bibfield  {title} {\bibinfo {title} {Catalysis in quantum
  information theory},\ }\href {https://doi.org/10.1103/RevModPhys.96.025005}
  {\bibfield  {journal} {\bibinfo  {journal} {Rev. Mod. Phys.}\ }\textbf
  {\bibinfo {volume} {96}},\ \bibinfo {pages} {025005} (\bibinfo {year}
  {2024})},\ \Eprint {https://arxiv.org/abs/2306.00798} {arXiv:2306.00798
  [quant-ph]} \BibitemShut {NoStop}%
\bibitem [{\citenamefont {Datta}\ \emph {et~al.}(2023)\citenamefont {Datta},
  \citenamefont {Varun~Kondra}, \citenamefont {Miller},\ and\ \citenamefont
  {Streltsov}}]{review2}%
  \BibitemOpen
  \bibfield  {author} {\bibinfo {author} {\bibfnamefont {C.}~\bibnamefont
  {Datta}}, \bibinfo {author} {\bibfnamefont {T.}~\bibnamefont {Varun~Kondra}},
  \bibinfo {author} {\bibfnamefont {M.}~\bibnamefont {Miller}},\ and\ \bibinfo
  {author} {\bibfnamefont {A.}~\bibnamefont {Streltsov}},\ }\bibfield  {title}
  {\bibinfo {title} {Catalysis of entanglement and other quantum resources},\
  }\href {https://doi.org/10.1088/1361-6633/acfbec} {\bibfield  {journal}
  {\bibinfo  {journal} {Reports on Progress in Physics}\ }\textbf {\bibinfo
  {volume} {86}},\ \bibinfo {pages} {116002} (\bibinfo {year} {2023})},\
  \Eprint {https://arxiv.org/abs/2207.05694} {arXiv:2207.05694 [quant-ph]}
  \BibitemShut {NoStop}%
\bibitem [{\citenamefont {Jonathan}\ and\ \citenamefont
  {Plenio}(1999)}]{Jonathan_1999}%
  \BibitemOpen
  \bibfield  {author} {\bibinfo {author} {\bibfnamefont {D.}~\bibnamefont
  {Jonathan}}\ and\ \bibinfo {author} {\bibfnamefont {M.~B.}\ \bibnamefont
  {Plenio}},\ }\bibfield  {title} {\bibinfo {title} {Entanglement-assisted
  local manipulation of pure quantum states},\ }\href
  {https://doi.org/10.1103/physrevlett.83.3566} {\bibfield  {journal} {\bibinfo
   {journal} {Phys. Rev. Lett.}\ }\textbf {\bibinfo {volume} {83}},\ \bibinfo
  {pages} {3566–3569} (\bibinfo {year} {1999})},\ \Eprint
  {https://arxiv.org/abs/quant-ph/9905071} {arXiv:quant-ph/9905071}
  \BibitemShut {NoStop}%
\bibitem [{\citenamefont {van Dam}\ and\ \citenamefont
  {Hayden}(2003)}]{vanDam2003}%
  \BibitemOpen
  \bibfield  {author} {\bibinfo {author} {\bibfnamefont {W.}~\bibnamefont {van
  Dam}}\ and\ \bibinfo {author} {\bibfnamefont {P.}~\bibnamefont {Hayden}},\
  }\bibfield  {title} {\bibinfo {title} {Universal entanglement transformations
  without communication},\ }\href {https://doi.org/10.1103/PhysRevA.67.060302}
  {\bibfield  {journal} {\bibinfo  {journal} {Phys. Rev. A}\ }\textbf {\bibinfo
  {volume} {67}},\ \bibinfo {pages} {060302} (\bibinfo {year} {2003})},\
  \Eprint {https://arxiv.org/abs/quant-ph/0201041} {arXiv:quant-ph/0201041}
  \BibitemShut {NoStop}%
\bibitem [{\citenamefont {Turgut}(2007)}]{turgut2007catalytic}%
  \BibitemOpen
  \bibfield  {author} {\bibinfo {author} {\bibfnamefont {S.}~\bibnamefont
  {Turgut}},\ }\bibfield  {title} {\bibinfo {title} {Catalytic transformations
  for bipartite pure states},\ }\href
  {https://doi.org/10.1088/1751-8113/40/40/012} {\bibfield  {journal} {\bibinfo
   {journal} {J. Phys. A Math. Theor.}\ }\textbf {\bibinfo {volume} {40}},\
  \bibinfo {pages} {12185} (\bibinfo {year} {2007})},\ \Eprint
  {https://arxiv.org/abs/0706.3654} {arXiv:0706.3654 [quant-ph]} \BibitemShut
  {NoStop}%
\bibitem [{\citenamefont {Brandão}\ \emph {et~al.}(2015)\citenamefont
  {Brandão}, \citenamefont {Horodecki}, \citenamefont {Ng}, \citenamefont
  {Oppenheim},\ and\ \citenamefont {Wehner}}]{Brand_o_2015}%
  \BibitemOpen
  \bibfield  {author} {\bibinfo {author} {\bibfnamefont {F.}~\bibnamefont
  {Brandão}}, \bibinfo {author} {\bibfnamefont {M.}~\bibnamefont {Horodecki}},
  \bibinfo {author} {\bibfnamefont {N.}~\bibnamefont {Ng}}, \bibinfo {author}
  {\bibfnamefont {J.}~\bibnamefont {Oppenheim}},\ and\ \bibinfo {author}
  {\bibfnamefont {S.}~\bibnamefont {Wehner}},\ }\bibfield  {title} {\bibinfo
  {title} {The second laws of quantum thermodynamics},\ }\href
  {https://doi.org/10.1073/pnas.1411728112} {\bibfield  {journal} {\bibinfo
  {journal} {PNAS}\ }\textbf {\bibinfo {volume} {112}},\ \bibinfo {pages}
  {3275–3279} (\bibinfo {year} {2015})},\ \Eprint
  {https://arxiv.org/abs/1305.5278} {arXiv:1305.5278 [quant-ph]} \BibitemShut
  {NoStop}%
\bibitem [{\citenamefont {Ng}\ \emph {et~al.}(2015)\citenamefont {Ng},
  \citenamefont {Mančinska}, \citenamefont {Cirstoiu}, \citenamefont
  {Eisert},\ and\ \citenamefont {Wehner}}]{Ng_2015}%
  \BibitemOpen
  \bibfield  {author} {\bibinfo {author} {\bibfnamefont {N.~H.~Y.}\
  \bibnamefont {Ng}}, \bibinfo {author} {\bibfnamefont {L.}~\bibnamefont
  {Mančinska}}, \bibinfo {author} {\bibfnamefont {C.}~\bibnamefont
  {Cirstoiu}}, \bibinfo {author} {\bibfnamefont {J.}~\bibnamefont {Eisert}},\
  and\ \bibinfo {author} {\bibfnamefont {S.}~\bibnamefont {Wehner}},\
  }\bibfield  {title} {\bibinfo {title} {Limits to catalysis in quantum
  thermodynamics},\ }\href {https://doi.org/10.1088/1367-2630/17/8/085004}
  {\bibfield  {journal} {\bibinfo  {journal} {New J. Phys.}\ }\textbf {\bibinfo
  {volume} {17}},\ \bibinfo {pages} {085004} (\bibinfo {year} {2015})},\
  \Eprint {https://arxiv.org/abs/1405.3039} {arXiv:1405.3039 [quant-ph]}
  \BibitemShut {NoStop}%
\bibitem [{\citenamefont {Wilming}\ and\ \citenamefont
  {Gallego}(2017)}]{Wilming2017}%
  \BibitemOpen
  \bibfield  {author} {\bibinfo {author} {\bibfnamefont {H.}~\bibnamefont
  {Wilming}}\ and\ \bibinfo {author} {\bibfnamefont {R.}~\bibnamefont
  {Gallego}},\ }\bibfield  {title} {\bibinfo {title} {{Third Law of
  Thermodynamics as a Single Inequality}},\ }\href
  {https://doi.org/10.1103/PhysRevX.7.041033} {\bibfield  {journal} {\bibinfo
  {journal} {Phys. Rev. X}\ }\textbf {\bibinfo {volume} {7}},\ \bibinfo {pages}
  {041033} (\bibinfo {year} {2017})},\ \Eprint
  {https://arxiv.org/abs/1701.07478} {arXiv:1701.07478 [quant-ph]} \BibitemShut
  {NoStop}%
\bibitem [{\citenamefont {\AA{}berg}(2014)}]{Aberg2014}%
  \BibitemOpen
  \bibfield  {author} {\bibinfo {author} {\bibfnamefont {J.}~\bibnamefont
  {\AA{}berg}},\ }\bibfield  {title} {\bibinfo {title} {Catalytic coherence},\
  }\href {https://doi.org/10.1103/PhysRevLett.113.150402} {\bibfield  {journal}
  {\bibinfo  {journal} {Phys. Rev. Lett.}\ }\textbf {\bibinfo {volume} {113}},\
  \bibinfo {pages} {150402} (\bibinfo {year} {2014})},\ \Eprint
  {https://arxiv.org/abs/1304.1060} {arXiv:1304.1060 [quant-ph]} \BibitemShut
  {NoStop}%
\bibitem [{\citenamefont {Vaccaro}\ \emph {et~al.}(2018)\citenamefont
  {Vaccaro}, \citenamefont {Croke},\ and\ \citenamefont
  {Barnett}}]{Vaccaro2018}%
  \BibitemOpen
  \bibfield  {author} {\bibinfo {author} {\bibfnamefont {J.~A.}\ \bibnamefont
  {Vaccaro}}, \bibinfo {author} {\bibfnamefont {S.}~\bibnamefont {Croke}},\
  and\ \bibinfo {author} {\bibfnamefont {S.~M.}\ \bibnamefont {Barnett}},\
  }\bibfield  {title} {\bibinfo {title} {Is coherence catalytic?},\ }\href
  {https://doi.org/10.1088/1751-8121/aac112} {\bibfield  {journal} {\bibinfo
  {journal} {J. Phys. A Math. Theor.}\ }\textbf {\bibinfo {volume} {51}},\
  \bibinfo {pages} {414008} (\bibinfo {year} {2018})},\ \Eprint
  {https://arxiv.org/abs/1804.05154} {arXiv:1804.05154 [quant-ph]} \BibitemShut
  {NoStop}%
\bibitem [{\citenamefont {Lostaglio}\ and\ \citenamefont
  {M{\"u}ller}(2019)}]{lostaglio2019coherence}%
  \BibitemOpen
  \bibfield  {author} {\bibinfo {author} {\bibfnamefont {M.}~\bibnamefont
  {Lostaglio}}\ and\ \bibinfo {author} {\bibfnamefont {M.~P.}\ \bibnamefont
  {M{\"u}ller}},\ }\bibfield  {title} {\bibinfo {title} {Coherence and
  asymmetry cannot be broadcast},\ }\href
  {https://doi.org/10.1103/PhysRevLett.123.020403} {\bibfield  {journal}
  {\bibinfo  {journal} {Phys. Rev. Lett.}\ }\textbf {\bibinfo {volume} {123}},\
  \bibinfo {pages} {020403} (\bibinfo {year} {2019})},\ \Eprint
  {https://arxiv.org/abs/1812.08214} {arXiv:1812.08214 [quant-ph]} \BibitemShut
  {NoStop}%
\bibitem [{\citenamefont {Takagi}\ and\ \citenamefont
  {Shiraishi}(2022)}]{takagi2022correlation}%
  \BibitemOpen
  \bibfield  {author} {\bibinfo {author} {\bibfnamefont {R.}~\bibnamefont
  {Takagi}}\ and\ \bibinfo {author} {\bibfnamefont {N.}~\bibnamefont
  {Shiraishi}},\ }\bibfield  {title} {\bibinfo {title} {Correlation in
  catalysts enables arbitrary manipulation of quantum coherence},\ }\href
  {https://doi.org/10.1103/PhysRevLett.128.240501} {\bibfield  {journal}
  {\bibinfo  {journal} {Phys. Rev. Lett.}\ }\textbf {\bibinfo {volume} {128}},\
  \bibinfo {pages} {240501} (\bibinfo {year} {2022})},\ \Eprint
  {https://arxiv.org/abs/2106.12592} {arXiv:2106.12592 [quant-ph]} \BibitemShut
  {NoStop}%
\bibitem [{\citenamefont {Char}\ \emph {et~al.}(2023)\citenamefont {Char},
  \citenamefont {Chakraborty}, \citenamefont {Bhar}, \citenamefont
  {Chattopadhyay},\ and\ \citenamefont {Sarkar}}]{char2023catalytic}%
  \BibitemOpen
  \bibfield  {author} {\bibinfo {author} {\bibfnamefont {P.}~\bibnamefont
  {Char}}, \bibinfo {author} {\bibfnamefont {D.}~\bibnamefont {Chakraborty}},
  \bibinfo {author} {\bibfnamefont {A.}~\bibnamefont {Bhar}}, \bibinfo {author}
  {\bibfnamefont {I.}~\bibnamefont {Chattopadhyay}},\ and\ \bibinfo {author}
  {\bibfnamefont {D.}~\bibnamefont {Sarkar}},\ }\bibfield  {title} {\bibinfo
  {title} {Catalytic transformations in coherence theory},\ }\href
  {https://doi.org/10.1103/PhysRevA.107.012404} {\bibfield  {journal} {\bibinfo
   {journal} {Phys. Rev. A.}\ }\textbf {\bibinfo {volume} {107}},\ \bibinfo
  {pages} {012404} (\bibinfo {year} {2023})},\ \Eprint
  {https://arxiv.org/abs/0908.1583} {arXiv:0908.1583 [quant-ph]} \BibitemShut
  {NoStop}%
\bibitem [{\citenamefont {Duan}\ \emph {et~al.}(2005)\citenamefont {Duan},
  \citenamefont {Feng}, \citenamefont {Li},\ and\ \citenamefont
  {Ying}}]{duan2005multiple}%
  \BibitemOpen
  \bibfield  {author} {\bibinfo {author} {\bibfnamefont {R.}~\bibnamefont
  {Duan}}, \bibinfo {author} {\bibfnamefont {Y.}~\bibnamefont {Feng}}, \bibinfo
  {author} {\bibfnamefont {X.}~\bibnamefont {Li}},\ and\ \bibinfo {author}
  {\bibfnamefont {M.}~\bibnamefont {Ying}},\ }\bibfield  {title} {\bibinfo
  {title} {Multiple-copy entanglement transformation and entanglement
  catalysis},\ }\href {https://doi.org/10.1103/PhysRevA.71.042319} {\bibfield
  {journal} {\bibinfo  {journal} {Phys. Rev. A}\ }\textbf {\bibinfo {volume}
  {71}},\ \bibinfo {pages} {042319} (\bibinfo {year} {2005})},\ \Eprint
  {https://arxiv.org/abs/quant-ph/0404148} {arXiv:quant-ph/0404148}
  \BibitemShut {NoStop}%
\bibitem [{\citenamefont {Lipka-Bartosik}\ and\ \citenamefont
  {Skrzypczyk}(2021)}]{PhysRevLett.127.080502}%
  \BibitemOpen
  \bibfield  {author} {\bibinfo {author} {\bibfnamefont {P.}~\bibnamefont
  {Lipka-Bartosik}}\ and\ \bibinfo {author} {\bibfnamefont {P.}~\bibnamefont
  {Skrzypczyk}},\ }\bibfield  {title} {\bibinfo {title} {Catalytic quantum
  teleportation},\ }\href {https://doi.org/10.1103/PhysRevLett.127.080502}
  {\bibfield  {journal} {\bibinfo  {journal} {Phys. Rev. Lett.}\ }\textbf
  {\bibinfo {volume} {127}},\ \bibinfo {pages} {080502} (\bibinfo {year}
  {2021})},\ \Eprint {https://arxiv.org/abs/2102.11846} {arXiv:2102.11846
  [quant-ph]} \BibitemShut {NoStop}%
\bibitem [{\citenamefont {Ganardi}\ \emph {et~al.}(2024)\citenamefont
  {Ganardi}, \citenamefont {Kondra},\ and\ \citenamefont
  {Streltsov}}]{Ganardi_2024}%
  \BibitemOpen
  \bibfield  {author} {\bibinfo {author} {\bibfnamefont {R.}~\bibnamefont
  {Ganardi}}, \bibinfo {author} {\bibfnamefont {T.~V.}\ \bibnamefont
  {Kondra}},\ and\ \bibinfo {author} {\bibfnamefont {A.}~\bibnamefont
  {Streltsov}},\ }\bibfield  {title} {\bibinfo {title} {Catalytic and
  asymptotic equivalence for quantum entanglement},\ }\href
  {https://doi.org/10.1103/physrevlett.133.250201} {\bibfield  {journal}
  {\bibinfo  {journal} {Phys. Rev. Lett.}\ }\textbf {\bibinfo {volume} {133}},\
  \bibinfo {pages} {250201} (\bibinfo {year} {2024})},\ \Eprint
  {https://arxiv.org/abs/2305.03488} {arXiv:2305.03488 [quant-ph]} \BibitemShut
  {NoStop}%
\bibitem [{\citenamefont {Wilming}(2021)}]{wilming2021entropy}%
  \BibitemOpen
  \bibfield  {author} {\bibinfo {author} {\bibfnamefont {H.}~\bibnamefont
  {Wilming}},\ }\bibfield  {title} {\bibinfo {title} {Entropy and reversible
  catalysis},\ }\href {https://doi.org/10.1103/PhysRevLett.127.260402}
  {\bibfield  {journal} {\bibinfo  {journal} {Phys. Rev. Lett.}\ }\textbf
  {\bibinfo {volume} {127}},\ \bibinfo {pages} {260402} (\bibinfo {year}
  {2021})},\ \Eprint {https://arxiv.org/abs/2012.05573} {arXiv:2012.05573
  [quant-ph]} \BibitemShut {NoStop}%
\bibitem [{\citenamefont {{Bavaresco et al.}}(2025)}]{inprep}%
  \BibitemOpen
  \bibfield  {author} {\bibinfo {author} {\bibfnamefont {J.}~\bibnamefont
  {{Bavaresco et al.}}},\ }\href@noop {} {\bibinfo {title} {in preparation}}
  (\bibinfo {year} {2025})\BibitemShut {NoStop}%
\bibitem [{\citenamefont {Horodecki}\ \emph {et~al.}(1999)\citenamefont
  {Horodecki}, \citenamefont {Horodecki},\ and\ \citenamefont
  {Horodecki}}]{PhysRevA.60.1888}%
  \BibitemOpen
  \bibfield  {author} {\bibinfo {author} {\bibfnamefont {M.}~\bibnamefont
  {Horodecki}}, \bibinfo {author} {\bibfnamefont {P.}~\bibnamefont
  {Horodecki}},\ and\ \bibinfo {author} {\bibfnamefont {R.}~\bibnamefont
  {Horodecki}},\ }\bibfield  {title} {\bibinfo {title} {General teleportation
  channel, singlet fraction, and quasidistillation},\ }\href
  {https://doi.org/10.1103/PhysRevA.60.1888} {\bibfield  {journal} {\bibinfo
  {journal} {Phys. Rev. A}\ }\textbf {\bibinfo {volume} {60}},\ \bibinfo
  {pages} {1888} (\bibinfo {year} {1999})},\ \Eprint
  {https://arxiv.org/abs/quant-ph/9807091} {arXiv:quant-ph/9807091 [quant-ph]}
  \BibitemShut {NoStop}%
\bibitem [{\citenamefont {Zhang}\ and\ \citenamefont
  {Chitambar}(2024)}]{zhang24exact}%
  \BibitemOpen
  \bibfield  {author} {\bibinfo {author} {\bibfnamefont {Y.}~\bibnamefont
  {Zhang}}\ and\ \bibinfo {author} {\bibfnamefont {E.}~\bibnamefont
  {Chitambar}},\ }\bibfield  {title} {\bibinfo {title} {Exact steering bound
  for two-qubit werner states},\ }\href
  {https://doi.org/10.1103/PhysRevLett.132.250201} {\bibfield  {journal}
  {\bibinfo  {journal} {Phys. Rev. Lett.}\ }\textbf {\bibinfo {volume} {132}},\
  \bibinfo {pages} {250201} (\bibinfo {year} {2024})},\ \Eprint
  {https://arxiv.org/abs/2309.09960} {arXiv:2309.09960 [quant-ph]} \BibitemShut
  {NoStop}%
\bibitem [{\citenamefont {Renner}(2024)}]{renner24compatibility}%
  \BibitemOpen
  \bibfield  {author} {\bibinfo {author} {\bibfnamefont {M.~J.}\ \bibnamefont
  {Renner}},\ }\bibfield  {title} {\bibinfo {title} {Compatibility of
  generalized noisy qubit measurements},\ }\href
  {https://doi.org/10.1103/PhysRevLett.132.250202} {\bibfield  {journal}
  {\bibinfo  {journal} {Phys. Rev. Lett.}\ }\textbf {\bibinfo {volume} {132}},\
  \bibinfo {pages} {250202} (\bibinfo {year} {2024})},\ \Eprint
  {https://arxiv.org/abs/2309.12290} {arXiv:2309.12290 [quant-ph]} \BibitemShut
  {NoStop}%
\bibitem [{\citenamefont {Clauser}\ \emph {et~al.}(1969)\citenamefont
  {Clauser}, \citenamefont {Horne}, \citenamefont {Shimony},\ and\
  \citenamefont {Holt}}]{chsh}%
  \BibitemOpen
  \bibfield  {author} {\bibinfo {author} {\bibfnamefont {J.~F.}\ \bibnamefont
  {Clauser}}, \bibinfo {author} {\bibfnamefont {M.~A.}\ \bibnamefont {Horne}},
  \bibinfo {author} {\bibfnamefont {A.}~\bibnamefont {Shimony}},\ and\ \bibinfo
  {author} {\bibfnamefont {R.~A.}\ \bibnamefont {Holt}},\ }\bibfield  {title}
  {\bibinfo {title} {Proposed experiment to test local hidden-variable
  theories},\ }\href {https://doi.org/10.1103/PhysRevLett.23.880} {\bibfield
  {journal} {\bibinfo  {journal} {Phys. Rev. Lett.}\ }\textbf {\bibinfo
  {volume} {23}},\ \bibinfo {pages} {880} (\bibinfo {year} {1969})}\BibitemShut
  {NoStop}%
\bibitem [{\citenamefont {Arnon-Friedman}\ \emph {et~al.}(2018)\citenamefont
  {Arnon-Friedman}, \citenamefont {Dupuis}, \citenamefont {Fawzi},
  \citenamefont {Renner},\ and\ \citenamefont {Vidick}}]{ArnonFriedman2018}%
  \BibitemOpen
  \bibfield  {author} {\bibinfo {author} {\bibfnamefont {R.}~\bibnamefont
  {Arnon-Friedman}}, \bibinfo {author} {\bibfnamefont {F.}~\bibnamefont
  {Dupuis}}, \bibinfo {author} {\bibfnamefont {O.}~\bibnamefont {Fawzi}},
  \bibinfo {author} {\bibfnamefont {R.}~\bibnamefont {Renner}},\ and\ \bibinfo
  {author} {\bibfnamefont {T.}~\bibnamefont {Vidick}},\ }\bibfield  {title}
  {\bibinfo {title} {Practical device-independent quantum cryptography via
  entropy accumulation},\ }\href {https://doi.org/10.1038/s41467-017-02307-4}
  {\bibfield  {journal} {\bibinfo  {journal} {Nature Communications}\ }\textbf
  {\bibinfo {volume} {9}},\ \bibinfo {pages} {459} (\bibinfo {year}
  {2018})}\BibitemShut {NoStop}%
\bibitem [{\citenamefont {Nadlinger}\ \emph {et~al.}(2022)\citenamefont
  {Nadlinger}, \citenamefont {Drmota}, \citenamefont {Nichol}, \citenamefont
  {Araneda}, \citenamefont {Main}, \citenamefont {Srinivas}, \citenamefont
  {Lucas}, \citenamefont {Ballance}, \citenamefont {Ivanov}, \citenamefont
  {Tan}, \citenamefont {Sekatski}, \citenamefont {Urbanke}, \citenamefont
  {Renner}, \citenamefont {Sangouard},\ and\ \citenamefont
  {Bancal}}]{Nadlinger2022}%
  \BibitemOpen
  \bibfield  {author} {\bibinfo {author} {\bibfnamefont {D.~P.}\ \bibnamefont
  {Nadlinger}}, \bibinfo {author} {\bibfnamefont {P.}~\bibnamefont {Drmota}},
  \bibinfo {author} {\bibfnamefont {B.~C.}\ \bibnamefont {Nichol}}, \bibinfo
  {author} {\bibfnamefont {G.}~\bibnamefont {Araneda}}, \bibinfo {author}
  {\bibfnamefont {D.}~\bibnamefont {Main}}, \bibinfo {author} {\bibfnamefont
  {R.}~\bibnamefont {Srinivas}}, \bibinfo {author} {\bibfnamefont {D.~M.}\
  \bibnamefont {Lucas}}, \bibinfo {author} {\bibfnamefont {C.~J.}\ \bibnamefont
  {Ballance}}, \bibinfo {author} {\bibfnamefont {K.}~\bibnamefont {Ivanov}},
  \bibinfo {author} {\bibfnamefont {E.~Y.-Z.}\ \bibnamefont {Tan}}, \bibinfo
  {author} {\bibfnamefont {P.}~\bibnamefont {Sekatski}}, \bibinfo {author}
  {\bibfnamefont {R.~L.}\ \bibnamefont {Urbanke}}, \bibinfo {author}
  {\bibfnamefont {R.}~\bibnamefont {Renner}}, \bibinfo {author} {\bibfnamefont
  {N.}~\bibnamefont {Sangouard}},\ and\ \bibinfo {author} {\bibfnamefont
  {J.-D.}\ \bibnamefont {Bancal}},\ }\bibfield  {title} {\bibinfo {title}
  {{Experimental quantum key distribution certified by Bell’s theorem}},\
  }\href {https://doi.org/10.1038/s41586-022-04941-5} {\bibfield  {journal}
  {\bibinfo  {journal} {Nature}\ }\textbf {\bibinfo {volume} {607}},\ \bibinfo
  {pages} {682–686} (\bibinfo {year} {2022})},\ \Eprint
  {https://arxiv.org/abs/2109.14600} {arXiv:2109.14600 [quant-ph]} \BibitemShut
  {NoStop}%
\bibitem [{\citenamefont {Navascu\'es}\ and\ \citenamefont
  {V\'ertesi}(2011)}]{Navascues2011}%
  \BibitemOpen
  \bibfield  {author} {\bibinfo {author} {\bibfnamefont {M.}~\bibnamefont
  {Navascu\'es}}\ and\ \bibinfo {author} {\bibfnamefont {T.}~\bibnamefont
  {V\'ertesi}},\ }\bibfield  {title} {\bibinfo {title} {Activation of nonlocal
  quantum resources},\ }\href {https://doi.org/10.1103/PhysRevLett.106.060403}
  {\bibfield  {journal} {\bibinfo  {journal} {Phys. Rev. Lett.}\ }\textbf
  {\bibinfo {volume} {106}},\ \bibinfo {pages} {060403} (\bibinfo {year}
  {2011})},\ \Eprint {https://arxiv.org/abs/1010.5191} {arXiv:1010.5191
  [quant-ph]} \BibitemShut {NoStop}%
\bibitem [{\citenamefont {Karvonen}(2021)}]{Karvonen2021}%
  \BibitemOpen
  \bibfield  {author} {\bibinfo {author} {\bibfnamefont {M.}~\bibnamefont
  {Karvonen}},\ }\bibfield  {title} {\bibinfo {title} {Neither contextuality
  nor nonlocality admits catalysts},\ }\href
  {https://doi.org/10.1103/PhysRevLett.127.160402} {\bibfield  {journal}
  {\bibinfo  {journal} {Phys. Rev. Lett.}\ }\textbf {\bibinfo {volume} {127}},\
  \bibinfo {pages} {160402} (\bibinfo {year} {2021})},\ \Eprint
  {https://arxiv.org/abs/2102.07637} {arXiv:2102.07637 [quant-ph]} \BibitemShut
  {NoStop}%
\bibitem [{\citenamefont {Hirsch}\ \emph {et~al.}(2016)\citenamefont {Hirsch},
  \citenamefont {Quintino}, \citenamefont {Bowles}, \citenamefont {Vértesi},\
  and\ \citenamefont {Brunner}}]{Hirsch2016}%
  \BibitemOpen
  \bibfield  {author} {\bibinfo {author} {\bibfnamefont {F.}~\bibnamefont
  {Hirsch}}, \bibinfo {author} {\bibfnamefont {M.~T.}\ \bibnamefont
  {Quintino}}, \bibinfo {author} {\bibfnamefont {J.}~\bibnamefont {Bowles}},
  \bibinfo {author} {\bibfnamefont {T.}~\bibnamefont {Vértesi}},\ and\
  \bibinfo {author} {\bibfnamefont {N.}~\bibnamefont {Brunner}},\ }\bibfield
  {title} {\bibinfo {title} {Entanglement without hidden nonlocality},\ }\href
  {https://doi.org/10.1088/1367-2630/18/11/113019} {\bibfield  {journal}
  {\bibinfo  {journal} {New Journal of Physics}\ }\textbf {\bibinfo {volume}
  {18}},\ \bibinfo {pages} {113019} (\bibinfo {year} {2016})},\ \Eprint
  {https://arxiv.org/abs/1606.02215} {arXiv:1606.02215 [quant-ph]} \BibitemShut
  {NoStop}%
\bibitem [{\citenamefont {T\'oth}\ \emph {et~al.}(2020)\citenamefont {T\'oth},
  \citenamefont {V\'ertesi}, \citenamefont {Horodecki},\ and\ \citenamefont
  {Horodecki}}]{Toth2020}%
  \BibitemOpen
  \bibfield  {author} {\bibinfo {author} {\bibfnamefont {G.}~\bibnamefont
  {T\'oth}}, \bibinfo {author} {\bibfnamefont {T.}~\bibnamefont {V\'ertesi}},
  \bibinfo {author} {\bibfnamefont {P.}~\bibnamefont {Horodecki}},\ and\
  \bibinfo {author} {\bibfnamefont {R.}~\bibnamefont {Horodecki}},\ }\bibfield
  {title} {\bibinfo {title} {Activating hidden metrological usefulness},\
  }\href {https://doi.org/10.1103/PhysRevLett.125.020402} {\bibfield  {journal}
  {\bibinfo  {journal} {Phys. Rev. Lett.}\ }\textbf {\bibinfo {volume} {125}},\
  \bibinfo {pages} {020402} (\bibinfo {year} {2020})}\BibitemShut {NoStop}%
\bibitem [{\citenamefont {Tr\'enyi}\ \emph {et~al.}(2024)\citenamefont
  {Tr\'enyi}, \citenamefont {Luk\'acs}, \citenamefont {Horodecki},
  \citenamefont {Horodecki}, \citenamefont {V\'ertesi},\ and\ \citenamefont
  {T\'oth}}]{Trenyi2024}%
  \BibitemOpen
  \bibfield  {author} {\bibinfo {author} {\bibfnamefont {R.}~\bibnamefont
  {Tr\'enyi}}, \bibinfo {author} {\bibfnamefont {A.}~\bibnamefont {Luk\'acs}},
  \bibinfo {author} {\bibfnamefont {P.}~\bibnamefont {Horodecki}}, \bibinfo
  {author} {\bibfnamefont {R.}~\bibnamefont {Horodecki}}, \bibinfo {author}
  {\bibfnamefont {T.}~\bibnamefont {V\'ertesi}},\ and\ \bibinfo {author}
  {\bibfnamefont {G.}~\bibnamefont {T\'oth}},\ }\bibfield  {title} {\bibinfo
  {title} {Activation of metrologically useful genuine multipartite
  entanglement},\ }\href {https://doi.org/10.1088/1367-2630/ad1e93} {\bibfield
  {journal} {\bibinfo  {journal} {New Journal of Physics}\ }\textbf {\bibinfo
  {volume} {26}},\ \bibinfo {pages} {023034} (\bibinfo {year} {2024})},\
  \Eprint {https://arxiv.org/abs/2203.05538} {arXiv:2203.05538 [quant-ph]}
  \BibitemShut {NoStop}%
\end{thebibliography}%

\appendix
\onecolumngrid

\hspace*{2cm}

\section*{APPENDIX}

\section{Proof of Lemma~\ref{lem: cata}}\label{app:proof}

Here we restate Lemma~\ref{lem: cata} from the main text and present its proof for the general case.

\setcounter{lemma}{0}

\begin{lemma}
    By means of local operations any state $\rho_{AB}$ can be catalytically transformed to 
    \begin{equation}
        \tau_{A' B'} =  \frac{1}{n}\, \rho_{AB}^{\otimes n} \otimes [00]_{R_A R_B} + \frac{n-1}{n}\, \sigma_{AB}^{\otimes n} \otimes [11]_{R_A R_B}
    \end{equation}
    where $A' \equiv A_1\dots A_n R_A$ and $B' \equiv B_1\dots B_n R_B$ are composed of $n\in \mathbb{N}$ copies of the system $A$ and $B$ and classical bits $R_A$ and $R_B$, respectively,
    and where $\sigma_{AB}=\sigma_A \otimes \sigma_B$ is an arbitrary product state.
\end{lemma}

\begin{proof}
    We start by defining the bipartite catalyst. To do so introduce the systems $C_A \equiv \widetilde A_1 \dots \widetilde A_{n-1} \widetilde R_A$ and 
    $C_B\equiv \widetilde B_1 \dots \widetilde B_{n-1} \widetilde R_B$ consisting of $n-1$ copies of the quantum system $A$ (respectively $B$)  denoted $\widetilde A_i$ (respectively $\widetilde B_i$) as well as a classical register $\widetilde R_A$ (respectively $\widetilde R_B$). Let the state of the catalyst be
    \begin{align}\label{eq: catalyst general}
        \omega_{C_A C_B} := \frac{1}{n}  \sum_{i = 0}^{n-1} (\rho^{\otimes i} \ot \sigma^{\otimes (n-1-i)})_{\widetilde A_1 \widetilde B_1 \dots \widetilde A_{n-1} \widetilde B_{n-1} } \ot[ii]_{\widetilde R_A \widetilde R_B}.
    \end{align}
    where each branch of the mixture the first $i$ pairs of systems $\widetilde A_k \widetilde B_k$ are in the state $\rho_{\widetilde A_k \widetilde B_k}$ (for $1\leq k\leq i$) and the remaining ones are in the state $\sigma_{\widetilde A_k \widetilde B_k}$ (for $i+1\leq k\leq n-1$).

    Now we define the local operations performing the desired transformation on the state $\rho_{AB}\ot \omega_{\widetilde A \widetilde B}$. The first step consists of reading out the value of the correlated registers $\widetilde R_A$ and $\widetilde R_B$, obtaining the values  $i\in\{0,\dots, n-1\}$ with uniform probability $\frac{1}{n}$. 
    
    If the observed value is $i=n-1$, the system $AC_A=A\widetilde A_1\dots \widetilde A_{n-1}$ and $B C_B=B\widetilde B_1\dots \widetilde B_{n-1}$ carry $n$ copies of the input state $\rho_{AB}$. Alice and Bob then locally swap these copies to the output systems $A_1\dots A_n$ and $B_1\dots B_n$ respectively, set the output classical registers $R_A R_B$ to $[00]_{R_A R_B}$ and prepare the catalyst in the product state $(\sigma^{\otimes (n-1)} \ot [00])_{C_A C_B}$ (which can be done locally). For $i=n-1$, the global output state is thus
    \begin{equation}
        \tau_{A' B' C_A C_B}^{(n-1)} = (\rho^{\otimes n} \otimes [00])_{A' B'}\ot (\sigma^{\ot (n-1)} \ot [00])_{C_A C_B}.
    \end{equation}
    
    For all other measured values of the registers $i\leq n-2$, the parties first locally swap the input systems $AB$ received in the state $\rho_{AB}$ in the catalyst (into the systems $\widetilde A_{i+1}\widetilde B_{i+1}$ to be exact) So that it  now carries $i+1\leq n-
    1$ copies of $\rho_{AB}$, and sets the catalyst registers to $[(i+1)(i+1)]_{\widetilde R_A \widetilde R_B}$. Then, they prepare the output systems $A' B' $ in the product state $( \sigma^{\otimes n} \otimes [11])_{A'  B' }$, independent of $i$. Hence, for $i\leq n-2$, the global output state is
    \begin{align}
    \begin{split}
        \tau_{A' B' C_ A C_B}^{(i)} &= (\sigma^{\otimes n} \otimes [11])_{A'B'}\ot  
        (\rho^{\otimes (i+1)} \ot \sigma^{\otimes (n-2-i)} \ot[(i+1)(i+1)])_{C_A C_B}.
    \end{split}    
    \end{align}
    
    The final state resulting from these local operations is the mixture 
    \begin{equation}
    \tau_{A' B' C_A C_B} = \frac{1}{n}\sum_{i=0}^{n-1} \tau_{A' B' C_A C_B}^{(i)}.
    \end{equation} 
    Is is straightforward to see that its marginals satisfy $\tau_{C_A C_B} \coloneqq \tr_{A' B'}(\tau_{A' B' C_A C_B})=\omega_{C_A C_B}$ in Eq.~\eqref{eq: catalyst}, and $\tau_{A' B'}\coloneqq \tr_{C_A C_B}(\tau_{A' B'C_A C_B})$ given by Eq.~\eqref{eq: out catalyst n}, concluding the proof.
\end{proof}

\section{Instrument-based catalytic activation of Bell nonlocality}\label{app: intrumental}

Throughout the main text we consider the scenario for catalytic activation of Bell nonlocality depicted in Fig.~\ref{fig:1}(b). Here, local operations $\mathcal{E}_{AC_A}$ and $\mathcal{E}_{B C_B}$ (CPTP maps) on systems $AC_A$ and $BC_B$ produce the global output state
\begin{align}
    {\bf(b)}: \qquad  &\tau_{A' B' C_A C_B} = (\mathcal{E}_{AC_A}\otimes \mathcal{E}_{B C_B})[\rho_{AB}\ot \omega_{C_A C_B}] \\ &\text{such that}\qquad \nonumber
    \tr_{A' B' }  \tau_{A' B' C_A C_B} = \omega_{C_AC_B}.
\end{align}

We now discuss a different form of catalytic activation of nonlocality, depicted in Fig.~\ref{fig:1}(c), where the catalyst is only required to be returned after the measurement outcomes have been produced.
In this version, the local operations of Alice and Bob are described by quantum instruments 
$\{\mathcal{I}^{a|x}_{AC_A}\}$ and $\{\mathcal{I}^{b|y}_{BC_B}\}$. For the classical inputs $x,y$, the instruments output the global classical-quantum state
\begin{align}\label{eq: cat variant 2new}
{\bf(c)}: \qquad &\tau_{O_A O_B C_A C_B}^{(x,y)} =\sum_{a,b} p(ab|xy)\,\,  [ab]_{O_A O_B} \otimes \omega_{C_AC_B}^{(a,b,x,y)} \quad \text{with} \quad
 p(ab|xy)\, \omega_{C_AC_B}^{(a,b,x,y)} =(\mathcal{I}^{a|x}_{AC_A} \ot \mathcal{I}^{b|y}_{BC_B})[\rho_{AB} \ot \omega_{C_AC_B}] \nonumber \\
 &\text{such that}\quad
  {\bf(c1)},  \,{\bf(c2)}, \, \text{or}  \, {\bf(c3)} \quad \text{hold}, 
\end{align}
where $O_A$ with $O_B$ are classical registers storing the outputs $a$ and $b$. Here we can consider three different variants to impose that this transformation is catalytic, formally given by
\begin{align}
    {\bf(c1)}:& \qquad \omega_{C_AC_B}^{(a,b,x,y)} = \omega_{C_AC_B} \quad  \forall x,y,a,b \quad \Longleftrightarrow \quad \tau_{O_A O_B C_A C_B}^{(x,y)} = \left(\sum_{a,b} p(ab|xy)\,\,  [ab]_{O_A O_B} \right) \otimes \omega_{C_A C_B}, \\
     {\bf(c2)}:& \qquad \sum_{ab} p(ab|xy) \omega_{C_AC_B}^{(a,b,x,y)} = \tr_{O_A O_B} \tau_{O_A O_B C_A C_B}^{(x,y)} = \omega_{C_AC_B} \quad \forall x,y \\
      {\bf(c3)}:&\qquad  \sum_{abxy} p_X(x) p_Y(y) \, p(a,b|x,y)\, \omega_{C_AC_B}^{(a,b,x,y)} = \omega_{C_A C_B},
\end{align}
where $\{p_X(x)\}$ and $\{p_Y(y)\}$ are the probably distributions from which the inputs are sampled by the parties. In words, $\bf (c1)$ requires that the state of the catalyst is unchanged for all inputs and all outputs, implying that the final classical-quantum states are product. $\bf (c2)$ requires that the marginal state of the catalyst (after discarding the output registers) is unchanged for all possible inputs. Arguably, it is the closest of the three to $\bf(b)$. In turn, $\bf (c3)$ requires that the  catalyst is unchanged also when the registers storing the inputs are discarded. This requires the explicit introduction of the probabilities $\{p_X(x)\}$ and $\{p_Y(y)\}$ with which the inputs are sampled.

It is straightforward to see that $\textbf{(c1)}\implies \textbf{(c2)} \implies \textbf{(c3)}$, hence there is a clear hierarchy between these classes of catalytic transformation of the state $\rho_{AB}$:  transformations possible under $\textbf{(c3)}$ include all transformations possible under  $\textbf{(c2)}$ etc. In is also not difficult to see that $\textbf{(b)}\implies \textbf{(c2)} \implies \textbf{(c3)}$, in the sense that any catalytic nonlocality activation scenario $\textbf{(b)}$ can always be interpreted as a particular case of scenario $\textbf{(c2)}$. This is because the local transformations and measurements in $\textbf{(b)}$ can be combined together to define the instruments in $\textbf{(c)}$, where $\textbf{(c2)}$ will be automatically satisfied since the measurement setting $x,y$ are not causally related to the catalyst. In turn, the scenarios $\textbf{(b)}$ and $\textbf{(c1)}$ seem incomparable. It is an open question if nonlocality activation is even possible in $\textbf{(c1)}$.

Finally, one could also consider catalytic activation of nonlocality assisted by shared randomness, which by definition cannot generate nonlocality on its own. Here, in addition to the catalyst, the parties are given access to shared randomness, \textit{which does not have to be returned}, unlike the catalyst. In this setting shared randomness cannot be na\"ively absorbed in the catalyst, since general local operation can create correlations between the randomness resources and the catalyst, thus modifying their global state. Hence, in principle, shared randomness might open more possibilities for catalytic Bell nonlocality activation.

\end{document}